\documentclass[12pt,a4paper]{article}
\usepackage{amsfonts}
\usepackage{amsthm}
\usepackage{amsmath}
\usepackage{amssymb}
\usepackage{amscd}

\usepackage{t1enc}
\usepackage[mathscr]{eucal}
\usepackage{indentfirst}
\usepackage{graphicx}
\usepackage{subfig}
\usepackage{graphics}
\usepackage{pict2e}
\usepackage{epic}
\usepackage{natbib}

\numberwithin{equation}{section}
\usepackage[margin=2.9cm]{geometry}
\usepackage{epstopdf} 

\usepackage[ruled,vlined]{algorithm2e}
\usepackage{appendix}
\usepackage{xcolor}

\theoremstyle{plain}
\newtheorem{Th}{Theorem}[section]
\newtheorem{Lemma}[Th]{Lemma}

\newtheorem{Prop}[Th]{Proposition}

 \theoremstyle{definition}
\newtheorem{Def}[Th]{Definition}

\newtheorem{?}[Th]{Problem}

\newtheorem{Remark}[Th]{Remark}

\newcommand{\BP}{\mathbb P}
\newcommand{\BE}{\mathbb E}

\begin{document}
\allowdisplaybreaks
\title{Family-wise Error Rate Control with E-values \thanks{We thank Emmanuel Cand\`{e}s, David Siegmund, Yo Joong Choe, Rianne de Heide, Lasse Fischer, Nick Koning, Ruodu Wang, Jelle Goeman, and members of the e-Readers group for helpful comments. L.L. is grateful for the support of National Science Foundation grant DMS-2338464. W.H. is supported by an Achievement Rewards for College Scientists Fellowship.}}
\author{Will Hartog \thanks{Department of Statistics, Stanford University. Email: whartog@stanford.edu } \and Lihua Lei \thanks{Graduate School of Business and Department of Statistics, Stanford University. Email: lihualei@stanford.edu}}

 \maketitle

\begin{abstract} 

The closure principle is a standard tool for achieving strong family-wise error rate (FWER) control in multiple testing problems. We develop an e-value-based closed testing framework that inherits nice properties of e-values, which are common in settings of sequential hypothesis testing or universal inference for irregular parametric models. We prove that e-value-based closed testing strongly controls the post-hoc FWER in the static setting, and has stronger anytime-valid and always-valid FWER-controlling properties in the sequential setting. Furthermore, we extend the celebrated graphical approach for FWER control \citep{bretz2009graphical}, using the weighted average of e-values for the local test, a strictly more powerful approach than weighted Bonferroni local tests with inverse e-values as p-values. In general, the computational cost for closed testing can be exponential in the number of hypotheses. Although the computational shortcuts for the p-value-based graphical approach are not applicable, we develop an efficient polynomial-time algorithm using dynamic programming for e-value-based graphical approaches with any directed acyclic graph, and tailored algorithms for the e-Holm procedure previously studied by \cite{vovk2021values} and the e-Fallback procedure. 
\end{abstract}

\maketitle

\section{Introduction and Preliminaries}

\subsection{Family-wise error rate}

In analyses that test multiple hypotheses, it is necessary to consider the problem of multiple testing, and control for some notion of false rejection. One such approach is to control the Family-wise Error Rate (FWER), the probability of any false rejection, at a given level $\alpha\in(0,1)$. In particular, given null hypotheses $H_1,...,H_n$, with $\mathcal{H}_0$ denoting the subset of true nulls, an algorithm that produces the (random) rejection set $\mathcal{R}$ based on test statistics controls FWER at level $\alpha$ if 
\[\mathbb{P}(\mathcal{R}\cap \mathcal{H}_0\neq \emptyset)\le \alpha.\]
In the literature, this property is often referred to as strong FWER control, to distinguish from the weak form, which applies only when all null hypotheses are true. Since we will not address weak FWER control, except briefly in Section \ref{sec:post-hoc} to clarify existing work, we will adopt the simpler terminology throughout.

Controlling the FWER is crucial to maintain the reliability of statistical results, especially in  medicine, economics, and experimentation in online platforms. In medicine, clinical trials often compare the efficacy of several dosages or involve multiple endpoints, and FWER control ensures that no false positives undermine patient safety or lead to the adoption of ineffective therapies \citep{pocock1987analysis, bretz2009graphical, vickerstaff2019methods}. In economics, researchers often test multiple hypotheses defined by multiple  interventions, multiple demographic subgroups, and multiple policy outcomes at the same time; controlling FWER helps avoid erroneous conclusions that could misguide policy decisions \citep{romano2005stepwise, list2019multiple, viviano2024model}. Similarly, in online platforms, A/B testing is used to compare multiple product versions or features under different metrics. Controlling FWER minimizes the chance of falsely identifying an intervention as superior, avoiding shipping ineffective features \citep{johari2022always}. Recently, FWER control has been applied to areas beyond traditional multiple testing problems, such as distribution-free risk control for black-box machine learning algorithms \citep{angelopoulos2021learn}. 

\subsection{Closed testing}

Given null hypotheses $H_1,...,H_n$, a local test for index subset $I\subset[n]$, where $[n]$ is shorthand notation for the set $\{1, \ldots, n\}$, tests the hypothesis $H_I=\cap_{i\in I}H_i$ that all the nulls $H_i$ in $I$ are true. By convention, we call $H_1, \ldots, H_n$ elementary hypotheses.

Given a family of valid local tests $\{\phi_I\}_{I\subset[n]}$, meaning $\BP_{H_I}(\phi_I=1)\leq\alpha$ for all $I\subset[n]$, the closure principle (Theorem 1 of \cite{marcus1976closed}) for $\{\phi_I\}_{I\subset[n]}$ rejects $H_i$, the elementary hypothesis, if $\phi_I=1$ for all $I\ni i$. Any admissible procedure which controls FWER is a closed test \citep{sonnemann1988Vollständigkeitssätze}. The closure principle has also been extended to controlling more general error measures such as the false discovery proportion (FDP) \citep{goeman2011multiple, goeman2021only,fischer2024online} and false discovery rate (FDR) \citep{xu2026bringing}. 

In the p-value literature, the most common local test is weighted Bonferroni, which given p-values $p_1,...,p_n$, and a set of valid weights for any $I$, i.e. $\{w_i(I)\}$ with 
\begin{equation}\label{eq:weight_condition}
\sum_{i\in I}w_i(I)\leq1,
\end{equation}
rejects $H_I$, or equivalently computes $\phi_I=1$ iff $\min_{i\in I}\frac{p_i}{w_i(I)}\leq\alpha$, which is valid as 
\[\BP_{H_I}(\phi_I=1)=\BP_{H_I}(\cup_{i\in I}\{p_i\leq\alpha w_i(I)\})\leq\sum_{i\in I}\BP_{H_I}(p_i\leq\alpha w_i(I))\leq\sum_{i\in I}\alpha w_i(I)\leq\alpha,\]
by a union bound.

The challenge becomes simplifying the computation from the $2^n-1$ local tests to an efficient number. The literature includes many such procedures including Holm's procedure \citep{holm1979simple}, which uses the local test of unweighted Bonferroni with $w_i(I) = 1/|I|$.

A larger class of local tests, to which Holm's procedure belongs, is that of the graphical approach, introduced by \cite{bretz2009graphical}. It is a generic framework that unifies most widely-used FWER controlling procedures and allows the researcher to encode logical relationships among hypotheses. We will revisit the local tests in detail in Section \ref{sec:dag}. Holm's procedure is an example of the graphical approach, with the complete graph and equal weights. The Fallback procedure, introduced by \cite{wiens2005}, is an example of the graphical approach with a chain graph. In general, any graph can be used for the purposes of the graphical approach, where its closure can be efficiently computed, with a simple greedy algorithm for determining the rejection set based on the p-values.

\subsection{E-values}\label{sec:eval_intro}
E-values are a recent alternative to p-values as a tool for hypothesis testing and quantifying evidence against the null \citep{wasserman2020universal,grunwald2024safe,shafer2021testing, waudby2024estimating, wang2022false}. For a null hypothesis $H_0$, an e-value for $H_0$ is a realization of an e-variable $e$ which has the property that $\BE_{H_0}[e]\leq1$. By Markov's inequality, the test $\phi=1\{e\geq1/\alpha\}$ is valid at level $\alpha$, since by Markov's inequality $\BP(e\geq1/\alpha)\leq\alpha$.

E-values arise naturally in sequential settings, especially in always-valid inference, where the practitioner would like the flexibility of having valid inference across time. Inference is always-valid if one can choose the time with the best power and the method is still valid, whereas inference is anytime-valid if one can choose any stopping time, and the method is still valid. We discuss these concepts in the context of the family-wise error rate further in Section \ref{sec:always-valid}. This framework is very relevant for the setting of online data collection in the tech sector, including the setting discussed by \cite{johari2022always}. In these settings, the e-value is often defined through an e-process $(e_{t})_{t\ge 1}$, which is a stochastic process which at any filtration-adapted stopping time $\tau$ has the property that $e_\tau$ is an e-value, i.e. $\mathbb{E}[e_\tau]\leq1$ \citep{ramdas2020admissible}; see Section \ref{sec:sims} for examples. An equivalent definition \citep{ramdas2023gametheoretic} is that the process is bounded by a non-negative martingale with marginal expectation $1$, known as a test martingale \citep{shafer2011test}.

E-values also prove to be useful in universal inference for irregular parametric models \citep{wasserman2020universal, tse2022note, spector2023discussion, park2023robust}. In general, e-values offer more robustness to get valid inference in the presence of various concerns about data-dependent experimental decisions, such as a post-hoc choice of $\alpha$ \citep{grunwald2024beyond, hemerik2024choosing, koning2024posthocalphahypothesistesting}. E-values are also admissible for online closed testing where hypotheses arrive sequentially \citep{fischer2024online}; however this differs from our setting with a fixed set of hypotheses.

\subsection{Our contribution: e-value-based closed testing}\label{sec:algorithmic_contribution}

Our contribution is threefold. First, we formalize the general e-closed testing framework when the local test for hypothesis $H_I$ is given by an e-value $e_I$ (Section \ref{sec:e-value-closed-testing-intro}) and develop a class of weighted e-Bonferroni local tests (Section \ref{sec:weighted-e-bonf}). 

Previous work studies a special case of unweighted e-Bonferroni which yields the e-Holm procedure that we will discuss at length \citep{vovk2021values, vovk2023confidence, vovk2024true}, and an online analogue of e-closed testing \citep{fischer2024online}. Later work generalizes the closure principle to false discovery rate control \citep{xu2026bringing}. 

Second, we establish for generic e-closed testing post-hoc strong FWER guarantees in static and sequential settings (Section \ref{sec:post-hoc}), which does not hold for p-closed testing, and the always-valid strong FWER guarantee (Section \ref{sec:always-valid}). Our result generalizes the post-hoc validity \citep{grunwald2024beyond, koning2024posthocalphahypothesistesting}, anytime-validity \citep{koning2025sequentializing}, and always-validity \citep{howard2021time,
johari2022always,turner2023exact} for single hypothesis testing, and develops a class of procedures to control the post-hoc FWER introduced by \cite{koning2024posthocalphahypothesistesting}. With weighted e-Bonferroni local tests and independent elementary pseudo e-values defined in Section \ref{sec:pareto}, such as the maximum of e-processes, we show that the closure approximately controls FWER when the target level is small. Third, we develop a generic dynamic programming-based algorithm for e-DAG procedures (Section \ref{sec:dag}) and tailored algorithms for special graphs (Section \ref{sec:efall}). We also develop a rejection rule for e-Holm (Section \ref{sec:e-holm}), which is faster than the e-Holm adjusted e-values algorithm developed by \cite{vovk2021values} and \cite{vovk2023confidence}.

Specifically, we first study closed testing procedures where each local hypothesis $H_I$ is associated with an e-value $e_I$. We prove that generic e-value-based closed tests control multiple stricter forms of Type-I error than the FWER. In particular, they achieve the post-hoc FWER control introduced by \cite{koning2024posthocalphahypothesistesting}, which allows  researchers to choose the target FWER level after observing the data. By contrast, we show that p-value-based closed testing does not control FWER post-hoc in general. In sequential settings where local e-values $e_I$ are given by e-processes \citep[e.g.][]{ramdas2023gametheoretic}, the resulting closed test 
guarantees the post-hoc anytime-valid strong FWER control: with probability $1-\alpha$, no null hypothesis would be rejected if the test is applied at a stopping time. In addition, it achieves always-valid strong FWER control: with probability $1-\alpha$, no null hypothesis would ever be rejected if the test is applied at every time point. 

Next, following the p-value-based closed testing literature, we investigate weighted e-Bonferroni local tests where the local e-value $e_I$ is a weighted average of e-values corresponding to elementary hypotheses. We show that the weighted e-Bonferroni tests are provably more powerful than the weighted p-Bonferroni tests with the same weights and $1/e$ p-values, the only admissible derived p-values without distributional assumptions on the e-values \citep{wang2024admissiblewaymergingevalues}. As a result, the closure of weighted e-Bonferroni tests is guaranteed to reject at least as many hypotheses as its p-value-based counterpart with $1/e$ p-values. We also examine the sequential settings where each elementary hypothesis $H_i$ is associated with an e-process $e_{it}$. In these settings, Ville's inequality \citep{ville1939etude} for nonnegative martingales $e_t$ says that
\begin{align}\label{eq:villes}
\mathbb{P}\left(\sup_{t\geq0}e_t\geq\frac{1}{\alpha}\right)\leq\alpha,
\end{align}
which yields a more powerful p-value, known as the always-valid p-value, as the inverse of $\max_{t} e_{it}$, the global maximum of the e-process \citep{johari2022always}. While $\tilde{e}_i = \max_{t}e_{it}$ is not a valid e-value, we show that, if the e-processes associated with elementary hypotheses are independent, the e-value-based closed test applied to pseudo-e-values $\tilde{e}_i$ approximately controls FWER when the target level is small and dominates the corresponding p-value-based closed test applied to always-valid p-values $1/\tilde{e}_i$. In fact, $\tilde e_i=1/p_i$ for any independent p-values $p_i$ have this asymptotic approximate validity.

Lastly, we develop efficient algorithms for e-graphical approaches, defined as the closure of weighted e-Bonferroni tests with the same weights used by p-graphical approaches and formally introduced in Section \ref{sec:weighted-e-bonf}.
The p-graphical approaches can be computed in polynomial time thanks to the consonance property of the weighted p-Bonferroni tests under the carefully constructed weights  \citep{gabriel1969simultaneous}. The consonance property yields a sequential rejection algorithm with $O(n^2)$ computational complexity. However, our e-value-based local tests are not consonant in general, as we discuss later in Section \ref{sec:dag}. While the standard sequential rejection algorithm does not work, we develop efficient polynomial time algorithms for e-graphical approaches with arbitrary directed acyclic graphs (DAG) using dynamic programming (Section \ref{sec:dag}). The computational complexity is $O(n|\mathcal{E}|)$ in general, where $|\mathcal{E}|$ is the number of edges, but usually smaller. We also improve the algorithm for two special cases: e-Holm (Section \ref{sec:e-holm}) and e-Fallback (Section \ref{sec:efall}). The e-Holm procedure was proposed in the prior work of \cite{vovk2021values, vovk2023confidence, vovk2024true}, though we derive a different representation based on the effective cutoff, computed in $O(n)$ time, which enables a direct comparison between e-Holm and p-Holm. Moreover, the previous work does not discuss other graphical approaches which pose substantially greater algorithmic challenges than e-Holm.

These contributions are significant because the e-closed testing approach is more powerful when e-values are the inferential object of interest, for example in settings where post-hoc validity or always-validity are necessary, or in sequential settings where the valid p-values are obtained through e-processes. Furthermore, when the p-values for elementary hypotheses are independent, we prove that the closure of weighted e-Bonferroni applied to pseudo 1/p e-values approximately controls FWER when $\alpha$ is small (e.g., $0.05$). This is provably more powerful than the closure of weighted p-Bonferroni with the same weights. On the other hand, when stronger FWER guarantees are not needed and powerful dependent p-values can be constructed, such as in the setting with fixed sample sizes and parametric hypotheses, best available e-values could be much less powerful than best p-values. As a result, the gains from the better aggregation for e-closed testing generally does not outweigh the power loss from individual tests.

\section{Theory of closed testing with e-values}\label{sec:e-value-closed-testing}

\subsection{E-value-based closed testing}\label{sec:e-value-closed-testing-intro}

In general, we define an e-value-based closed testing procedure as a collection of e-values $\{e_I\}_{I\subset[n]}$ where each $e_I$ is a valid e-value for the intersection hypothesis $H_I=\cap_{i\in I}H_i$. The procedure rejects  $H_i$ iff $e_I\ge 1/\alpha$ for all supersets $I$ containing $i$. In this section, we make no structural assumption on the e-values. Later in Section \ref{sec:weighted-e-bonf}, we will consider a special case called weighted e-Bonferroni as the intersection e-value $e_I$. 

For p-value closed testing, it is convenient to report an adjusted p-value $p_i^*$ for each hypothesis $H_i$ such that $H_i$ is rejected by the closed test iff $p_i^* \le \alpha$. Analogously, we define the adjusted e-value 
\begin{equation}\label{eq:adjusted_evalue}
e_i^* \triangleq \min_{I\ni i}e_I.
\end{equation}
Clearly, $H_i$ is rejected by the e-value-based closed test iff $e_i^*\ge 1/\alpha$.

These definitions yield an immediate yet crucial result that lays the foundation of all error-controlling properties of e-value-based closed tests.
\begin{Lemma}\label{lem:key}
Let $e_i^*$ be adjusted e-values defined in \eqref{eq:adjusted_evalue}. For $\mathcal{H}_0=\{i:H_i\text{ is true}\}$ the set of true nulls,
\[\max_{i\in \mathcal{H}_0}e_i^*\le e_{\mathcal{H}_0}.\]
\end{Lemma}

Since $e_{\mathcal{H}_0}$ is a valid e-value, we can easily prove the FWER controlling property. 

\begin{Prop} \label{prop:standard_fwer}
The e-value-based closed test controls the strong FWER: 

\begin{align*}
    \mathrm{FWER} = \mathbb{P}\left(\max_{i\in \mathcal{H}_0}e_i^*\geq\frac{1}{\alpha} \right)\leq \mathbb{P}\left(e_{\mathcal{H}_0}\geq\frac{1}{\alpha} \right)\leq \alpha.
\end{align*}
\end{Prop}

Other notions of closed testing with e-values include \cite{fischer2024online}, which introduces closure and proves the necessity of e-values for admissibility in the online setting where hypotheses arrive sequentially for simultaneous FDP control, and \cite{xu2026bringing}, which extends closed testing to FDR control. \cite{xu2026bringing} show that their FDR-controlling procedure, which likewise relies on an e-value $e_I$ for each subset, can decide post-hoc to control FWER through Theorem 12, but it is tied to the simultaneous FDR control and therefore less powerful than our strictly FWER-focused context.

In the following two subsections, we leverage Lemma \ref{lem:key} to prove stronger error-controlling properties of e-value-based closed testing. 

\subsection{Post-hoc FWER control}\label{sec:post-hoc}

\cite{grunwald2024beyond} formalizes the notion of roving $\alpha$s, or post-hoc choice of $\alpha$ at which to control the error rate. They prove that if using an e-value $e$ to test a single hypothesis, one can interpret the data at any (potentially data-dependent) level since for any $\hat\alpha$ which is a function of the data, \begin{align*}
\BE\left[\frac{\mathbf{1}\{e\geq\frac{1}{\hat\alpha}\}}{\hat\alpha}\right]\leq \BE[e] \le 1,
\end{align*}
\cite{koning2024posthocalphahypothesistesting} generalizes the post-hoc type-I error control to multiple testing. In particular, the Appendix C.2 of  \cite{koning2024posthocalphahypothesistesting} introduces 
the notion of post-hoc \emph{weak} FWER control in terms of controlling the probability of a false rejection when all hypotheses of interest are nulls, i.e., $\mathcal{H}_0 = \{1, \ldots, n\}$.

We prove that the post-hoc \emph{strong} FWER control can be achieved by e-value-based closed testing.

\begin{Th}\label{th:post-hoc-fwer} For closed testing adjusted e-values $\{e_i^*\}_{i=1}^n$ and set of null hypotheses $\mathcal{H}_0$, for any data-dependent function of the data $\hat\alpha$,
    \begin{align*}
    &\BE\left[\frac{\mathbf{1}\{\max_{i\in\mathcal{H}_0} e_i^*\geq\frac{1}{\hat\alpha}\}}{\hat\alpha}\right]\leq1 
\end{align*}
\end{Th}

In general, a p-value defined by $\mathbb{P}[p\leq\alpha]\leq\alpha$ is not post-hoc valid, so we would expect a p-value-based closed test to not control post-hoc FWER. As a concrete example, consider a collection of independent p-values with $n_0$ uniform null p-values $p_i:i\in\mathcal{H}_0$. Assume that for the non-null p-values there is some $\epsilon>0$ for which $\mathbb{P}[p>\epsilon]>0$. Then the p-Holm procedure rejects $H_i$ iff its adjusted p-value $p_i^* \triangleq \max_{I\ni i}\min_{j\in I}|I|p_j$ is below $\alpha$. If we choose $\hat{\alpha} = \min_{i}p_i^*$, the minimal level at which at least one hypothesis is rejected, and define the minimum p-values in the null and non-null cases respectively as $p_{(1)}^0,p_{(1)}^1$, a sufficient condition for making a false rejection is that $p_{(1)}^0<p_{(1)}^1$, which is when the smallest p-value is a true null. Under this event, the smallest adjusted p-value is $np_{(1)}^0$. Then we can bound
\begin{align*}
\mathbb{E}\left[\frac{\mathbf{1}\{\min_{i\in\mathcal{H}_0}p_i^*\leq\hat{\alpha}\}}{\hat{\alpha}}\right]& \geq\mathbb{E}\left[\frac{\mathbf{1}\{p_{(1)}^0<p_{(1)}^1\}}{\min_{i\in\mathcal{H}_0}p_i^*}\right]\geq\mathbb{E}\left[\frac{\mathbf{1}\{p_{(1)}^0<\epsilon,p_{(1)}^1>\epsilon\}}{\min_{i\in\mathcal{H}_0}p_i^*}\right]\\
&\geq\mathbb{E}\left[\frac{\mathbf{1}\{p_{(1)}^0<\epsilon,p_{(1)}^1>\epsilon\}}{np_{(1)}^0}\right]\\&=\mathbb{P}[p_{(1)}^1>\epsilon]\cdot\int_0^\epsilon\frac{n_0(1-u)^{n_0}}{nu}du=\infty,
\end{align*}
where the first inequality is from $\min_{i\in \mathcal{H}_0}p_i^*\geq\hat\alpha$, and the third inequality is from $p_{(1)}^0\geq\min_{i\in\mathcal{H}_0}p_{(i)}$ and therefore $np_{(1)}^0\geq\min_{i\in\mathcal{H}_0}p_i^*$.

\subsection{Stronger FWER control in sequential settings}\label{sec:always-valid}

In this section, we consider e-value-based closed testing in the sequential setting. Let $\mathcal{F}_t$ denote the $\sigma$-field generated by data up to time $t$ collected for all hypotheses. For each intersection hypothesis $H_I$, we define an e-process $e_{It}$ with respect to this filtration $\mathcal{F}_t$. An e-process is defined as a sequential process $\{e_t\}_{t\geq0}$ such that at any filtration adapted stopping time $\tau$ its expectation is bounded by 1 under the null: $\mathbb{E}[e_\tau]\leq1$. A common e-process for sequential independent data is the sequential probability ratio test (e.g. \cite{johari2022always}), which we use in our simulations in Section \ref{sec:sims}. \cite{wang2025anytime} refer to e-processes of this kind as global e-processes, in contrast to local e-processes that are valid with respect to their own filtrations. Then we can also define the adjusted e-processes $e^*_{it}$ as the adjusted e-values at time $t$, i.e. $e_{it}^*=\min_{I\ni i}e_{It}$. By Lemma \ref{lem:key}, we have that for each $t$, \begin{align}\label{eq:max-nulls}
    \max_{i\in\mathcal{H}_0}e^*_{it}\leq e_{\mathcal{H}_0,t}.
\end{align} 
Since $e_{\mathcal{H}_0,t}$ is  an e-process, the optional stopping theorem implies that $e_{\mathcal{H}_0, \tau}$ is a valid e-value for any $\mathcal{F}_t$-stopping time $\tau$. Using the same argument as in Theorem \ref{th:post-hoc-fwer}, we prove the post-hoc-anytime-validity. 
\begin{Th}\label{th:post-hoc-anytime-valid}
For any $\mathcal{F}_t$-stopping time $\tau$ and $\mathcal{F}_\tau$ measurable level $\hat{\alpha}$, 
    \begin{align*} &\BE\left[\frac{\mathbf{1}\{\max_{i\in\mathcal{H}_0} e_{i\tau}^*\geq\frac{1}{\hat\alpha}\}}{\hat\alpha}\right]\leq1 
\end{align*}
\end{Th}

While anytime-validity enables more flexible timing for decisions, it forbids making decisions based on retrospective inspection. By Ville's inequality \eqref{eq:villes}, we can prove the always-validity of e-value-based closed testing with a fixed target level. 

\begin{Th}\label{th:always-valid}
    For any fixed level $\alpha \in (0, 1)$, 
    \begin{align*}  \BP\left(\max_{i\in\mathcal{H}_0}\sup_te_{it}^*\geq\frac{1}{\alpha}\right)\leq\alpha
    \end{align*}
\end{Th}

As a result of Theorem \ref{th:always-valid}, the strong FWER of any null adjusted e-value \textit{ever} being rejected is controlled by $\alpha$. This means that we can make rejections based on the supremum over $t$ over the adjusted e-values, even if the supremum was in the past. This property is stronger than anytime-validity; see \cite{ramdas2020admissible} for a detailed comparison between two notions of sequential validity.

\section{Closed testing with weighted e-Bonferroni}\label{sec:weighted-closed-testing}

\subsection{Weighted e-Bonferroni local test}\label{sec:weighted-e-bonf}

In this section, we introduce one method of obtaining local test e-values using a weighted average. This framework allows for more refined robustness results for family-wise error control, discussed in Section \ref{sec:pareto}.

To obtain a valid local test with individual e-values $\{e_i\}_{i=1}^n$, one approach is to convert them into p-values using e-to-p calibrators \citep{vovk2021values}. Vovk and Wang (2021, Proposition 2.2) show that the unique  admissible converted p-value is the inverse e-value truncated at $1$, i.e., $\min\{1/e_i, 1\}$, which is further formalized in \cite{wang2024admissiblewaymergingevalues}. We will ignore the truncation without loss of generality, as none of the procedures discussed in this paper reject p-values greater than or equal to $1$. Then the weighted Bonferroni test with weights defined by \eqref{eq:weight_condition} rejects $H_I$ iff 
\begin{equation}\label{eq:inverse_e_bonferroni}
\min_{i\in I}\frac{1}{w_i(I) e_i}\le \alpha \Longleftrightarrow \max_{i\in I} w_i(I)e_i\ge \frac{1}{\alpha}.
\end{equation}

While \eqref{eq:inverse_e_bonferroni} is a valid test, it does not fully unleash the potential of e-values. An
appealing feature of e-values is that they combine easily; the weighted average of e-values is an e-value, and \cite{vovk2021values} show that the unweighted mean is optimal among e-merging functions. It is also possible to combine
independent e-values by multiplication, which \cite{fischer2024online} use with
closed testing, but we do not make any independence between e-values assumptions.

This suggests an alternative local test, which we call \emph{weighted e-Bonferroni} due to its analogy to weighted p-Bonferroni, that rejects $H_I$ iff

\begin{equation}\label{eq:e_bonferroni}
e_I\triangleq \sum_{i\in I}w_i(I) e_i \ge \frac{1}{\alpha}.
\end{equation}
The test is valid as $e_I$ is a valid e-value. 

When the weights $w_i(I)$ are chosen according to a p-graphical approach, we call the closed test an \emph{e-graphical approach}. The nomenclature extends to specific graphical approaches, such as e-Holm and e-Fallback.

Comparing \eqref{eq:inverse_e_bonferroni} with \eqref{eq:e_bonferroni}, it is evident that the weighted e-Bonferroni test is more powerful in the sense that it rejects $H_I$ whenever the weighted p-Bonferroni test \eqref{eq:inverse_e_bonferroni} with inverse e-values does. As a consequence, the e-value-based closed test is at least as powerful as its inverse-e p-value counterpart. This result provides motivation for using a closed testing procedure based on e-Bonferroni rather than p-Bonferroni when doing multiple testing with e-values. It does not say anything directly about a comparison with p-Bonferroni based on other p-values.

The weighted e-Bonferroni local test, with weights $w_i(I)$, leads to the following optimization to compute adjusted e-values: 
\begin{equation}
e_i^*=\min_{I\ni i}\left\{\sum_{i\in I}w_i(I)e_i\right\}.
\end{equation}
\cite{vovk2021values, vovk2023confidence} proposed an efficient algorithm for e-Holm where $w_i(I) = 1/|I|$ in our notation. We develop dynamic programming-based polynomial-time algorithms for more general e-graphical approaches in Sections \ref{sec:efall} and \ref{sec:dag}.

\subsection{Working with independent pseudo e-values}\label{sec:pareto}

The weighted e-Bonferroni local test provides a framework to prove a more detailed characterization of the family-wise error rate when using the running maxima to make rejections.

An e-process can be interpreted through its max $\tilde e_i=\max_{t\geq0}e_{it}$, as by Ville's inequality \eqref{eq:villes}, the probability under the null of ever exceeding $1/\alpha$ is bounded by $\alpha$: $\mathbb{P}(\tilde e_i\geq1/\alpha)\leq \alpha$. We call $\tilde e_i$ a pseudo e-value, as it is valid for testing $H_i$ at threshold $1/\alpha$ but it is not an e-value, as $\mathbb{E}[\tilde e_i]>1$ unless $e_{it}\leq1$ for all $t$ a.s..

Therefore, in the multiple testing context, we cannot plug these pseudo e-values into our e-closed testing framework and get FWER control for free. However, the result of Theorem \ref{th:always-valid} is unsatisfactory as the adjustment at time $t$ ignores previous peaks that the e-processes have achieved. We cannot prove that e-closed testing with pseudo e-values (equivalently running maxima) is valid at level $\alpha$, but we can prove its FWER is close to $\alpha$ when $\alpha$ is small and the pseudo e-values are independent.

We can bound the rejection probability through the following result, where we define $\{\tilde e_i^*\}_{i=1}^n$ to be the set of pseudo e-values adjusted by e-closed testing with weighted e-Bonferroni. 

\begin{Th}\label{th:fwer-pareto-bound}
    For any $K\ge 1$, define   \begin{align*}
B_K(\alpha):=\mathbb{P}\left(\sum_{i=1}^KY_i\geq\frac{K}{\alpha}\right),
    \end{align*}
    where $Y_1, \ldots, Y_K$ are i.i.d. with $\mathbb{P}(Y_i \ge y)=1/y, \,\, y\ge 1$.
    
    For a collection of independent e-processes $\{e_{it}\}_{t\geq0,i\in[n]}$, and corresponding pseudo e-values $\{\tilde e_i\}_{i\in[n]}$, e-closed testing with weighted average e-values applied to the pseudo e-values controls the FWER at level $B_{n_0}(\alpha)$ for $n_0=|\mathcal{H}_0|$ where $\mathcal{H}_0$ is the set of true nulls:
    \begin{align*}
\mathbb{P}\left(\max_{i\in\mathcal{H}_0}\tilde e_i^*\geq\frac{1}{\alpha}\right)\leq B_{n_0}(\alpha)=\alpha+\frac{n_0-1}{n_0}\alpha^2\log\frac{1}{\alpha}+O_{n_0}(\alpha^2),
    \end{align*}
    where $O_{n_0}(\alpha^2)$ denotes a term that is bounded by $C_{n_0}\alpha^2$ for some constant that only depends on $n_0$.
As a consequence, the e-closed testing applied to running maxima $\tilde e_{it}^* := \max_{s\le t}e_{it}$ achieves always-valid FWER control at level $B_{n_0}(\alpha)$.
\end{Th}

This result follows from a coupling of the pseudo e-values with independent $\mathrm{Pareto}(1)$ random variables, and an exact recursion on the integral in Lemma \ref{lem:pareto-bound}, and in general corresponds to an asymptotic interpretation of having FWER $\alpha+o(\alpha)$. The properties of e-value-based procedures with $\alpha \rightarrow 0$ have also been studied in other contexts \citep{koning2024continuous, koning2025sequentializing}. 

We plot $B_K(0.05)$ for a wide range of $K$ in Figure \ref{fig:B_k-plot} and observe $B_n(0.05)$ is very close to $0.05$. Since most FWER guarantees tend to be loose in practice, we suggest applying the e-closed testing on $\tilde e_i$ directly without making corrections on $\alpha$. Note that this version is guaranteed to be more powerful than the corresponding p-closed-test applied to inverse maxima of e-processes. If exact FWER control is required, one can adjust the level to $B_n^{-1}(0.05)$, which is plotted in Figure \ref{fig:B_k-plot}. Again, the adjusted level is close to $0.05$. We compare both the unadjusted and adjusted tests  with the p-closed-test counterpart empirically in Section \ref{sec:fixed-sample}.

\begin{figure}
    \centering
    \includegraphics[width=0.5\linewidth]{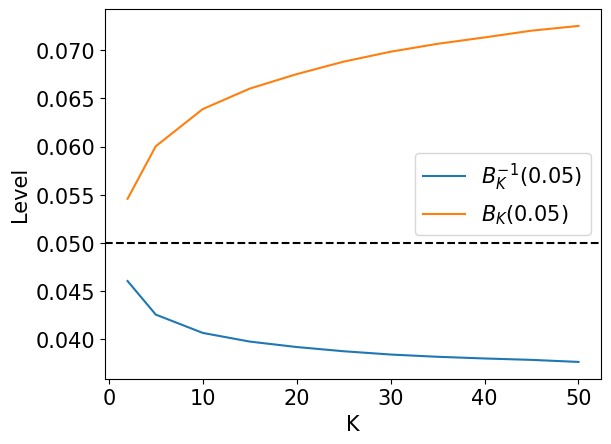}
    \caption{We show in orange the type-I error bound $B_K(\alpha)$ for $\alpha=0.05$ for varying sizes $K$, as well as in blue the value of $\alpha'=B_K^{-1}(\alpha)$ for which $B_K(\alpha')$ would be valid at level $\alpha$.}
    \label{fig:B_k-plot}
\end{figure}

\begin{Remark} Because the only requirement for these results was the Pareto dominance of $\tilde e_i$, it follows that any p-value $p_i$ can be transformed into a pseudo e-value $\tilde e_i=1/p_i$ and the approximate FWER control holds for the e-Holm procedure.
    
\end{Remark}

Extending post-hoc inference to be always-valid requires additional structure on the e-processes, as per \cite{tavyrikov2025carefree}. Following their work, we use calibrators to get post-hoc always-valid family-wise error rate control in Appendix \ref{sec:always-post-hoc}.

\section{E-Holm}\label{sec:e-holm}

\subsection{A simple expression of the e-Holm rejection rule}

In this section, we use the example of Holm's procedure to build a comparison between two strategies to control FWER given individual e-values $\{e_i\}_{i=1}^n$: a na\"{i}ve strategy based on calibrating the e-values into p-values, and the strategy based on e-closed testing with weighted average e-values introduced in Section \ref{sec:weighted-e-bonf}.

Holm's procedure is the closure of a set of unweighted p-Bonferroni local tests with weight $w_i(I) = 1/|I|$ for every subset $I\subset [n]$ and $i\in I$. Following the description in Section \ref{sec:weighted-e-bonf}, the e-Holm procedure closes the unweighted e-Bonferroni local tests: 
\[H_I \text{ rejected iff }\frac{1}{|I|}\sum_{i\in I}e_i\ge \frac{1}{\alpha}.\]

\begin{Remark}
By Lemma \ref{lem:key}, we only need $e_i$ to be compound e-values, first used in \cite{wang2022false} and coined by Definition 2.1 of \cite{ignatiadis2025asymptotic}, which are random variables whose expectations sum to at most $|I|$ for any subset $I$.
\end{Remark}

By definition of the closure principle, the e-Holm procedure rejects $H_i$ iff 

\begin{align}\label{eq:ebonf}
    \sum_{j\in I}e_j\geq\frac{|I|}{\alpha} \text{ for all } I\ni i.
\end{align}

Surprisingly, the rejection rule boils down to a simple thresholding rule with an easy-to-compute critical value, which can be computed in $O(n)$ time.

\begin{Th}\label{th:eholm_rejection} Let $J^*=\{j:e_j<1/\alpha\}$ be the set of insignificant e-values. Then the e-Holm procedure rejects $H_i$ iff 
\begin{equation}\label{eq:eholm_rejection}
e_{i} \ge \frac{1}{\alpha}+\sum_{j\in J^*}\left(\frac{1}{\alpha} - e_{j}\right).
\end{equation}
\end{Th}

Intuitively, when comparing a local e-value $e_I$ to the threshold $1/\alpha$, the worst case excludes all e-values with $e_j\geq1/\alpha$, as these will bring the average e-value closer to rejection. Then if the average e-value including $i$ and all insignificant e-values in $J^*$ is rejected, we know that all $I\ni i$ are rejected. The full proof of Theorem \ref{th:eholm_rejection} is in Appendix \ref{sec:proof-appendix}.

E-Holm was previously studied by \cite{vovk2021values}, who develop an $O(n\log n)$ adjusted e-values algorithm (2021, Algorithm 1) and prove its FWER control (2021, Theorem 5.1). Our $O(n)$ rejection set algorithm in Theorem \ref{th:eholm_rejection} is new. For completeness, we restate and explain the intuition behind the adjusted e-values algorithm in Appendix D, which is in the supplementary material.

Notably, the threshold rejection procedure \eqref{eq:eholm_rejection} can be computed without sorting the e-values, so Theorem \ref{th:eholm_rejection} does provide a runtime improvement over the adjusted e-values algorithm if the goal is to recover the rejection set for a fixed $\alpha$.

\subsection{Comparing e-Holm with p-Holm}\label{sec:e-holm-vs-p-holm} The p-Holm procedure with inverse-e p-values rejects the hypothesis $H_{(i)}$ corresponding to the $i$th largest e-value $e_{(i)}$ iff 
\[\frac{1}{e_{(j)}}\le \frac{\alpha}{n - j + 1}, \quad \text{for all }j\le i.\]
The discussion in Section \ref{sec:weighted-e-bonf} already implies that \eqref{eq:eholm_rejection} is less stringent. To bring more insight, we facilitate the comparison by considering an extreme case where the non-nulls are strong and $e_i > 1/\alpha$ almost surely (or with extremely high probability). In order for p-Holm to make any rejection, the largest e-value $e_{(1)}$ needs to exceed $n/\alpha$. By contrast, the threshold in e-Holm \eqref{eq:eholm_rejection} is bounded by $(|\mathcal{H}_0| + 1)/\alpha$ from above, which can be much smaller than $n/\alpha$. Therefore, e-Holm can adapt to the number of nulls while p-Holm cannot. This is analogous to adaptive false discovery rate control \citep{storey2004strong, benjamini2006adaptive}.

As a concrete example of this adaptivity, we consider a regime with very strong non-nulls in which, for increasing number of non-nulls, the power of e-Holm improves upon the power of p-Holm in a fixed sample size. The setting is with $n$ hypotheses where we observe $X_i$ and the null is $H_{0,i}:X_i\sim\mathrm{Unif}[0,1]$ and the point alternative is a mixture $H_{1,i}:X_i\sim\pi\mathrm{Unif}[0,1]+(1-\pi)\mathrm{Unif}[1-c,1]$, where $c\in(0,1)$ is some small constant. With probability $1-\pi$, under the alternative, the data is guaranteed to be close to 1.

In this case, the optimal Neyman-Pearson p-value relies on thresholding the likelihood ratio, which is 
\begin{align*}
    \mathrm{LR}(x)=\begin{cases}
        \pi+\frac{1-\pi}{c} & x \geq1-c \\ \pi & x < 1-c.
    \end{cases}
\end{align*}
The optimal p-value is non-unique because the likelihood ratio only takes two values. We show in Lemma \ref{lem:neyman-pearson} that $p(X_i) = 1-X_i$ is an optimal Neyman-Pearson p-value for $H_{0,i}$. 

On the other hand, we take as our e-value $e_i=\mathbf{1}\{X_i\geq1-c\}/c$, an all-or-nothing e-value in the vein of all-or-nothing e-values in \cite{shafer2011test,vovk2021values}. 
 
By Theorem \ref{th:eholm_rejection}, e-Holm makes a rejection iff 
\[\frac{1}{c}\ge \frac{n - M + 1}{\alpha}, \quad \text{where }M=\#\{i:X_i\geq1-c\}\Longleftrightarrow M\geq n+1-\alpha/c =: m_{c,\alpha},\]
in which case all hypotheses with $X_i\ge 1 - c$ are rejected. 

We present a simulation with $\alpha=0.05$, $\pi=0.01$, and $c=0.01$, with $n_1$ varying between $[5,10]$ out of $n=10$ total hypotheses. We show in Figure \ref{fig:unif-sim} that e-Holm becomes significantly more powerful than p-Holm for a large number of non-nulls, while maintaining FWER control. This confirms the heuristic discussion at the beginning of the section that e-Holm is more adaptive to the null proportion than p-Holm. 

\begin{figure}
    \centering
    {\includegraphics[width=0.48\textwidth]{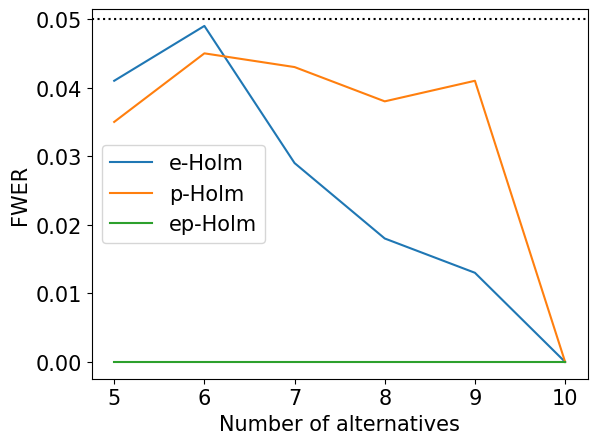}}
  \hfill
 {\includegraphics[width=0.48\textwidth]{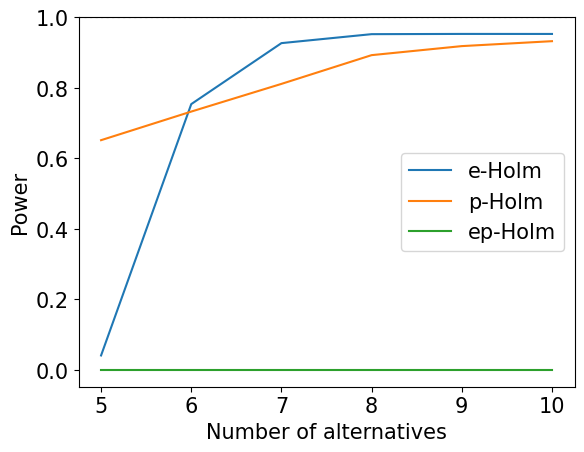}}
    \caption{Family-wise error rate and power for testing a uniform null hypothesis against a mixture uniform with probability $1-\pi$ of being above $1-c$, with $\pi=c=0.01$. The total number of hypotheses is $n=10$ and the number of non-nulls $n_1$ varies between 5 and 10. We compare e-Holm with p-Holm and ep-Holm, which is e-Holm applied to inverse-e p-values $\min(1,1/e_i)$.}
    \label{fig:unif-sim}
\end{figure}

\section{E-Fallback}\label{sec:efall}

\subsection{A naive quadratic time algorithm}
The Fallback procedure is defined by a fixed sequence of hypotheses, $H_1\rightarrow H_2\rightarrow\cdots\rightarrow H_n$ with associated e-values, and initial $\alpha$ budget $\{\alpha_i\}_{i=1}^n$ with $\sum_{i=1}^n \alpha_i=\alpha$. This is a special case of the later discussed graphical case with a chain graph. 

The Fallback procedure is useful when hypotheses have a pre-specified ordering of priority \citep{wiens2005}. The p-Fallback procedure iterates sequentially through the hypotheses, and if it can reject a p-value $p_i$ at the adjusted level that $\alpha$-budget will carry over to the next hypothesis, but disappear if it cannot be rejected. For example, if $p_1$ is rejected at level $\alpha_1$, then the rejection threshold for $p_2$ is $\alpha_1+\alpha_2$. Further, if $p_2$ is significant at this level, then the threshold for $p_3$ is $\alpha_1+\alpha_2+\alpha_3$. However, if $p_3$ is not rejected, then the threshold for $p_4$ is only $\alpha_4$.

For $|I|=k$ with $I=\{i_1,...,i_k\}$ and $i_1<i_2<\cdots<i_k$, the local test rejects $H_I$ iff
\begin{align}
e_I \triangleq \sum_{\ell=1}^k w_{i_\ell}(I)e_{i_\ell}\geq \frac{1}{\alpha}, \quad \text{where }w_{i_\ell}(I) = \sum_{i=i_{\ell-1}+1}^{i_\ell}\frac{\alpha_i}{\alpha}.
\end{align}
Above, we take $i_0=0$ for notational simplicity. This formula corresponds to redistributing $\alpha$ budget from excluded indices to the nearest subsequent included index.

We start by presenting a suboptimal yet simple algorithm to illustrate how we can apply dynamic programming to simplify the computation of the adjusted e-values $e_i^*=\min_{I\ni i}e_I$, since there is a chronological order between the hypothesis. We first prove that the adjusted e-values only depend on subsets with $i$ being the largest element. This result extends to the general DAGs (see Theorem \ref{th:dag-alg} in Section \ref{sec:dag}) and turns out to be a crucial property that enables an efficient dynamic programming approach. 
\begin{Lemma}\label{lem:remove_larger_than_i}
For each $i \in [n]$, 
\[e_i^{*} = \min_{I\ni i, i = \max_{j\in I} j}e_I.\]
\end{Lemma}

Based on this property, we give Algorithm \ref{alg:e-fallback-naive}, a dynamic programming algorithm with $O(n^2)$ runtime to compute the minimum over $I\ni i$ for each $i\in[n]$.

\begin{algorithm}
    \textit{Input:} Vectors $\{e_i\}_{i=1}^n,\{\alpha_i\}_{i=1}^n$;
    
    \textit{Initialize:} $m_0=0$;
    
    \textit{Iterate:} $m_i=\min_{0\leq j<i}\{m_j+e_i\sum_{k=j+1}^i\alpha_k\}$;
    
    \textit{Output:} $\{m_i\}_{i=1}^n$.

    \caption{E-Fallback}\label{alg:e-fallback-naive}
\end{algorithm}

We now prove that our computation of $m_i$ gives the adjusted e-values. 

\begin{Th}\label{th:efall-pf} The output of Algorithm \ref{alg:e-fallback-naive} has the property that $m_i = \alpha e_{i}^{*}$ for all $i \in [n]$. 
\end{Th}

This dynamic programming approach performs a minimization over at most $n$ values for each of $n$ indices so has $O(n^2)$ runtime. However, as we will see in the next two subsections, we can achieve an $O(n)$ runtime with some algorithmic shortcuts.

\subsection{Speedup via a reverse search algorithm}\label{sec:rev-search}

In this subsection, we show that there is an algorithm to perform the optimization of e-Fallback without searching all $i$ previous hypotheses each time we compute $m_i$. This algorithm will serve as a stepping stone to the graphical approaches with general DAGs in Section \ref{sec:dag}, as well as the tailored algorithm for e-Fallback in Section \ref{sec:stack}.

For any $i\in I\subset [i]$, instead of viewing $e_I$ as a weighted average of e-values, we can interpret it as a weighted average of initial $\alpha$-budgets $\alpha_1, \ldots,\alpha_i$, each appearing exactly once in the expression. For any such subset $I = \{i_1, \ldots, i_k\}$ with $0 = i_0 < i_1 < i_2 < \ldots < i_k = i$,  it divides $[i]$ into $k$ pieces and assigns the weight $e_{i_\ell}$ to the $\ell$-th piece $\{\alpha_{i_{\ell - 1}+1}, \ldots, \alpha_{i_{\ell}}\}$ in the formula of $e_I$. Adding an element $j\in (i_{\ell-1}, i_\ell)$ in $I$ amounts to cut the $\ell$-th chunk into two and reassigns the weight $e_j$ to $\alpha$-budgets in $\{i_{\ell-1}+1, \ldots, j\}$. To achieve the minimum, one should never add $j$ with $e_j > e_{i_{\ell}}$. This observation implies a recursive structure of the optimal subset for each hypothesis. 
\begin{Th}\label{th:efallback_optimal_I}
For each $i\in [n]$, let $j(i) = \max\{j < i: e_j \le e_i\}$ if the set is nonempty and $j(i) = 0$ otherwise. Define $I_i^*$ iteratively as follows: 
\[I_0^* = \emptyset, \quad I_{i}^{*} = I_{j(i)}^*\cup \{i\}.\]
Then 
\[e_{i}^{*} = e_{I_i^*} = \frac{1}{\alpha}\left(\sum_{j=j(i)+1}^i \alpha_j\right) e_i + e_{j(i)}^*.\]
\end{Th}

Theorem \ref{th:efallback_optimal_I} leads to the following reverse search algorithm. Again $\{m_i/\alpha\}_{i=1}^n$ is our set of adjusted e-values.

\begin{algorithm}
    \textit{Input:} Vectors $\{e_i\}_{i=1}^n,\{\alpha_i\}_{i=1}^n$;
    
    \textit{Initialize:} $m_i=0$ for all $i=0...n$;
    
    \textit{Iterate:} For $i=1$ to $n$:

    \hspace{0.9cm} \textit{Initialize:} $\alpha_i^*=\alpha_i$
    
    \hspace{0.9cm} \textit{Do:} Set $j=i-1$ and repeat until $e_j\leq e_i$ or $j=0$:

    \hspace{1.8cm} $\alpha_i^*=\alpha_i^*+\alpha_j$, $j=j-1$

    \hspace{0.9cm} $m_i=m_j+\alpha_i^*e_i$
    
    \textit{Output:} $\{m_i\}_{i=1}^n$.
    
    \caption{Reverse search for e-Fallback}\label{alg:rev-search}
\end{algorithm}

This algorithm gives a significant shortcut when the previous smaller e-value $e_j$ is close to $e_i$. If, for example, we assume that the e-values are continuous and exchangeable, then the expected runtime is $O(n)$. We state this as Theorem C.1, which we prove in Appendix C, which is in the supplementary material. However, in the worst case when the e-values are decreasing, the runtime, based on the total number of back searches, is still $O(n^2)$. In the next subsection, we will utilize a stack to recycle intermediate quantities more efficiently and reduce the computational complexity to $O(n)$.

\subsection{A tailored algorithm for e-Fallback}\label{sec:stack}
Inspecting the construction of the optimal subset $I_i^*$ in Theorem \ref{th:efallback_optimal_I}, we observe that 
\begin{equation}\label{eq:backward_cummin}
j \in I_i^* \Longleftrightarrow e_j = \min\{e_j, \ldots, e_i\},
\end{equation}
that is, $e_j$ is a cumulative minimum defined in a backward manner. $I_i^*$ is the lower envelope of the first $i$ e-values. We illustrate it in Figure \ref{fig:stack-vis}. This implies the following  iterative construction of $I_i^*$ in a similar spirit to Lemma D.1 for e-Holm. 

\begin{algorithm}
    \textit{Input:} Vectors $\{e_i\}_{i=1}^n,\{\alpha_i\}_{i=1}^n$;
    
    \textit{Initialize:} $m_i=0$ for all $i=0...n$, stack $s=\{\}$;
    
    \textit{Iterate:} For $i=1$ to $n$: 
    
    \hspace{0.9cm} \textit{Do:} Set $\alpha_i^*=\alpha_i$:

    \hspace{0.9cm} While $s\neq\{\}$ and $e_j>e_i$ for $j$ the top of stack:

    \hspace{1.8cm} Pop $(j,\alpha_j^*)$ from $s$ and add $\alpha_i^*=\alpha_i^*+\alpha_j^*$

    \hspace{0.9cm} If $s=\{\}$ set $m_i:=\alpha_i^*m_i$ else set $m_i:=\alpha_i^*e_i+m_j$ for $j$ the top of stack

    \hspace{0.9cm} Push $(i,\alpha_i^*)$ to $s$
    
    \textit{Output:} $\{m_i\}_{i=1}^n$.
    
    \caption{Stack search for e-Fallback}\label{alg:stack}
\end{algorithm}

\begin{Lemma}\label{lem:iterative_Ii*}
Suppose $I_i^*$ defined in Theorem \ref{th:efallback_optimal_I} is $\{k_{1}(i), \ldots, k_{\ell_i}(i)\}$ where $\ell_i = |I_i^*|$. Then $e_{k_1(i)}\le \ldots \le e_{k_{\ell_i}(i)}$ and 
\[I_{i}^{*} = \{i\}\cup I_{i-1}^{*}\setminus \{j\in I_{i-1}^{*}: e_j > e_i\}.\]
\end{Lemma}

The construction in Lemma \ref{lem:iterative_Ii*} can be implemented by a stack, as detailed in Algorithm \ref{alg:stack}.

\begin{figure}[htp]
    \centering
    \includegraphics[scale = 0.8]{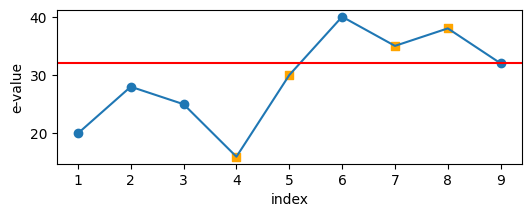}
    \caption{Visualization of stack e-Fallback Algorithm \ref{alg:stack}. For $H_8$, $I_8^*$ includes all four orange squares. At step $i=9$, the 7th and 8th squares are kicked out because their e-values are above $e_9$ and hence are no longer cumulative minima. The 4th and 5th squares are kept because $e_5$ is the most recent square whose e-value is below $e_9$. Going into step 10, $e_9$ will be added to the stack.}
    \label{fig:stack-vis}
\end{figure}

As desired, Algorithm \ref{alg:stack} has runtime $O(n)$. To see this, we break down the computation into push and pop operations. Since at most two elements are pushed at the end of each iteration, the total cost of push operations is $O(n)$. Using the representation \eqref{eq:backward_cummin}, all but the last popped elements at any iteration will not be added back to the stack again. The cost of last pop adds up to $O(n)$, while the cost of all other pops combined cannot exceed the number of hypotheses $n$. Therefore, the overall computational complexity is $O(n)$.

The existence of an $O(n)$ algorithm for e-Fallback is notable, because of the lack of the consonance that makes p-Fallback easy and intuitive to compute. A natural interpretation of p-Fallback is that a practitioner allocates their $\alpha$-budget $\{\alpha_i\}$, and each time they make a rejection, that budget accumulates and carries over to the next hypothesis, but is lost once a hypothesis cannot be rejected. 

The explicit recursive construction of the minimizing index set in Lemma \ref{lem:iterative_Ii*} is the key to the intuition for e-Fallback, and has a similar but more complicated interpretation. Since e-Fallback funnels budget to the smallest possible e-values, subject to the ordering, it carries over budget when e-values are larger than the subsequent e-values, instead of when they are rejected, as in p-Fallback. In an algorithmic sense, e-Fallback keeps budget in a more flexible manner than p-Fallback. As in our discussion of e-Holm in Section \ref{sec:e-holm-vs-p-holm}, this comparison is most applicable when comparing e-Fallback to the inverse-e p-Fallback, whereas p-Fallback with other p-values differs in other ways.

\section{E-graphical approach for general DAGs}\label{sec:dag}

\subsection{Revisiting the p-graphical approach}

A p-graphical approach \citep{bretz2009graphical} takes as input an initial $\alpha$-budget $\alpha_i$ for $H_i$, with $\sum_{i}\alpha_i=\alpha$ and a weighted directed graph $\mathcal{G} = (\mathcal{V}, \mathcal{E})$, with $\mathcal{V} = \{H_1, \ldots, H_n\}$. Each vertex encodes a hypothesis and the weight $q_{jk}$ between $H_j$ and $H_k$ is nonnegative and satisfying $\sum_{k}q_{jk}\le 1$, where $q_{jk}>0$ if and only if $(j,k)\in\mathcal{E}$. Throughout the paper we assume the graph $\mathcal{G}$ is given. It encodes logical constraints or relative importance of the hypotheses that depend on the domain knowledge. We refer to \cite{bretz2009graphical} for guidance on the choice of the graph and common examples in practice. 

To define the weights $w_i(I)$ for the local test $H_I$, expand the graph with an additional vertex, denoted by $H_{n+1}$, with graph weight $q_{j(n+1)} = 1 - \sum_{k}q_{jk}$ and $q_{(n+1)j} = \mathbf{1}\{j = n+1\}$.  Consider the random walk $Z_0, Z_1, \ldots$ on the expanded graph with
\[\mathbb{P}(Z_0 = H_i) = \frac{\alpha_i}{\alpha}, \quad \mathbb{P}(Z_{t+1} = H_k\mid Z_{t} = H_j) = q_{jk}.\]
Then, the weight $w_i(I)$ is chosen as 
\begin{equation}\label{eq:weight_graphical}
w_i(I)= \mathbb{P}(Z_{t_I} = H_i)
\end{equation}
where $t_I = \min\{t: Z_t\in \{H_i: i\in I\}\}$. This definition of weight redistribution is equivalent to sequentially removing vertices and redistributing the removed vertices' weight according to the graph transition probabilities. This derivation comes from a recursive absorption probability computation \citep{sheskin1991computing}.

As mentioned in Section \ref{sec:algorithmic_contribution}, the closure of the above local tests can be computed in polynomial time due to the consonance property \citep{gabriel1969simultaneous, romano2011consonance}, which states that, for any $I \subset [n]$, $H_I$ is rejected if there exists $i\in I$ such that $H_J$ is rejected for any $i\in J \subset I$. To see this, $w_i(I)\le w_i(J)$ for any $J \subset I$, as any such index $i$ included in both sets cannot lose weight when other indices are removed from $J$. When $H_I$ is rejected, there must be an $i\in I$ such that $p_i\le \alpha w_i(I)$. Then, for any $i\in J \subset I$, $p_i \le \alpha w_i(J)$  and hence $H_J$ is rejected. The graph is allowed to be cyclical. 

However, the weighted e-Bonferroni local tests are not consonant. For example, in the context of e-Holm, with three hypotheses, desired level $\alpha=0.05$, and e-values $(e_1,e_2,e_3)=(25, 25, 10)$, the subset $\{1,2,3\}$ is rejected as the average e-value is $e_{\{1,2,3\}}=20\geq\frac{1}{\alpha}$. However, there is no $i$ that satisfies the definition of the consonance property; in particular, $e_{\{1,3\}}=e_{\{2,3\}}=17.5<\frac{1}{\alpha}$.

This difference in the consonance property means that it is more difficult to compute the closure of e-Bonferroni under a graphical weight structure, which we call the e-graphical approach. Some cyclical graphs have appropriate structure for a shortcut, such as the complete graph with equal initial weights, which is the e-Holm procedure from Section \ref{sec:e-holm}, and graphs which are suitably close to acyclical, which we discuss in the appendix.

In this section, we focus on a broad class of directed acyclical graphs (DAGs), of which the Fallback graph from Section \ref{sec:efall} is an example. These graphs have an inherent shortcoming because excluded vertices may distribute their weight to the additional $H_{n+1}$ vertex, but they can be more interpretable, especially in scenarios in which a primary node has to precede secondary nodes, such as in the factorial design example in Section \ref{sec:sims}. Note, however, that one can always compute the e-graphical approach by brute force for any graph.

\subsection{Dynamic programming for e-DAG graphical approaches}
For a DAG, the hitting probabilities can be simplified using the unidirectional flow of the graph. To facilitate the computation, for any subset $I$ define the set of paths from $j\notin I$ to $i\in I$ through the complement of $I$: \begin{equation}
    \mathcal{P}_{j,i}^{(I)}=\{(i_0,...,i_k):k\geq0,i_0=j,i_k=i,i_\ell\notin I\,\,\forall\ell\in[k-1],q_{i_{\ell-1},i_\ell}>0\,\,\forall\ell\in[k]\}.
\end{equation}

Note the positivity of a graph weight ($q_{i_{\ell-1},i_\ell}>0$) indicates the presence of an edge. Implicit in this definition is that $i\in I,j\notin I$. Also define $k(p):=k$ the length of a path $p=(i_0,...,i_k)\in\mathcal{P}_{j,i}^{(I)}$. Then,

\begin{equation}\label{eq:weight_DAG}
w_i(I) = \frac{1}{\alpha}\left(\alpha_i+\sum_{j\notin I}\alpha_j\sum_{p\in\mathcal{P}_{j,i}^{(I)}}\prod_{\ell=1}^{k(p)}q_{i_{\ell-1},i_\ell}\right).
\end{equation}
We may think of the local test weight corresponding to redistributing the weight ($\alpha_j/\alpha$) from each excluded vertex $j$ proportionally according to the graph structure. With the weight $w_i(I)$ defined in \eqref{eq:weight_DAG}, the local test for e-DAG graphical approaches rejects $H_I$ iff
\begin{equation}\label{eq:eDAG_local_test}
e_I \triangleq \sum_{i\in I}w_i(I)e_i\geq\frac{1}{\alpha}.
\end{equation}

For a general DAG, we notice a local phenomenon that we claim will extend to an algorithm that computes the global minimum for any index $i$. We first establish an equivalence between the local test e-value $e_I$ as a weighting of e-values $e_i$ for $i\in I$, and a weighting of alpha values $\alpha_j$ for $j\in[n]$. 

\begin{Def}\label{def:e_j^I} For any subset $I\subset[n]$ and hypothesis index $j\in[n]$, define \begin{align}
    e_j^{(I)}\triangleq\begin{cases}
        \sum_{i\in I}e_i\left(\sum_{p\in\mathcal{P}_{j,i}^{(I)}}\prod_{\ell=1}^{k(p)}q_{i_{\ell-1},i_\ell}\right) & j\notin I \\
        e_j & j \in I
    \end{cases}
\end{align}
When no path exists, we take $e_j^{(I)} = 0$.
\end{Def}

Intuitively, we give this definition to give structure to the way in which an excluded vertex $i\notin I$ will have its alpha-budget $\alpha_j$ redistributed to the e-values in $I$ via a reweighting based on the graph structure. Formally, we prove this relationship in Lemma \ref{lem:e_j^I}.

\begin{Lemma}\label{lem:e_j^I} For $e_j^{(I)}$ as defined in Definition \ref{def:e_j^I} and subset $I\subset[n]$, we have
\begin{align}
    e_I=\sum_{i\in I}w_i(I)e_i=\frac{1}{\alpha}\sum_{j=1}^n\alpha_je_j^{(I)}
\end{align}
\end{Lemma}

The result of Lemma \ref{lem:e_j^I} is useful as it shows that if we can find a subset $I\ni i$ that minimizes $e_j^{(I)}$ for all $j$, then we have found the minimizing subset, $\mathrm{argmin}_{I\ni i}e_I$.

To facilitate the upcoming proof of our algorithm's validity, we also show the following Lemma, which says that we may recursively define $e_j^{(I)}$ by its children's values.

\begin{Lemma}\label{lem:e_j^I-rec}

For $j\notin I$,
\begin{align*}
    e_j^{(I)}=\sum_{(j,k)\in E}q_{j,k}e_k^{(I)}
\end{align*}
\end{Lemma}

We assume that the nodes are already topologically ordered with respect to the input graph, which is a reasonable constraint given in practice a practitioner will be designing the graph themselves so will be able to design the data collection to have the nodes ordered properly. 

We give the extension of the reverse search Fallback algorithm from Section \ref{sec:rev-search} to the general DAG in Algorithm \ref{alg:dag}. The intuition is that we can look solely at the ancestor graph, and search backward to decide whether or not to include each node, thereby reconstructing the minimizing subset for each node $i$.

\begin{algorithm}
    \textit{Input:} Vectors $\{e_i\}_{i=1}^n,\{\alpha_i\}_{i=1}^n$, transition matrix $(q_{ij})\in\mathbb{R}^{n\times n}$, where $[n]$ is a \\
    \hspace{1.2cm} topological ordering with respect to the graph $\mathcal{G}=(\mathcal{V},\mathcal{E})$;

    \textit{Iterate:} For $i\in[n]$:

    \hspace{0.9cm} \textit{Do:} Initialize e-assignments $e_j^{(i)} = 0\,\, (j\in[n])$, $m_i=0$, compute $A_i$ the  \\
    \hspace{1.8cm} set of ancestors of $i$, including $i$;

    \hspace{0.9cm} Initialize $e_i^{(i)}=e_i$, $m_i=\alpha_i e_i$:
    
    \hspace{0.9cm} For $i\neq j\in A_i$ decreasing $e_j^{(i)}=\min\left( e_j,\sum_{(j,k)\in E,k\in A_i}q_{jk}e_k^{(i)}\right)$ and set
    
    \hspace{0.9cm} $m_i=m_i+\alpha_je_j^{(i)}$.

    \textit{Output:} $\{m_i\}_{i=1}^n$
    
    \caption{e-DAG graphical approach}\label{alg:dag}
\end{algorithm}

\begin{Th}\label{th:dag-alg} In the context of Algorithm \ref{alg:dag}, for every index $i$ and subset $I\ni i$, $e_j^{(I)}\geq e_j^{(i)}$. Moreover, $e_j^{(i)}=e_j^{(I_i^*)}$ for some $I_i^*\ni i$, which is therefore the minimizing subset.
As a result, $e_i^* = m_i / \alpha$.    
\end{Th}

In Algorithm \ref{alg:dag}, for each node $i$, we recompute the e-value assignments for all $j\leq i$, each of which is a comparison of $e_j$ to a weighting of its children's assignments. The total computational complexity is $O\left( \sum_i N_i\right)$ where $N_i$ is the number of the edges in the ancestor graph $A_i$. This yields a worst-case complexity bound $O(n|\mathcal{E}|) = O(n^3)$ because $N_i \le |\mathcal{E}| = O(n^2)$. This worst-case bound can be improved for special graphs. For example, when $\mathcal{G}$ is a tree graph, $N_i$ is simply the depth of node $H_i$. For a balanced binary tree, $\max_i N_i = O(\log n)$, and thus the complexity is $O(n \log n) <\!\!< O(n|\mathcal{E}|)$. It may be possible to further improve the complexity by re-using intermediate calculations like Algorithm \ref{alg:stack}. For example, in the tree setting, each reverse search is on a chain, so each reverse search can use the shortcut from Algorithm \ref{alg:rev-search}, and in particular we believe one can use a depth-first search along with a stack, following Algorithm \ref{alg:stack}, to achieve an $O(n)$ runtime. We leave further developments for future research. 

\section{Simulation results}\label{sec:sims}

\subsection{Sequential setting}

In this section we develop a framework for comparing the power of our e-value procedures to that of the corresponding p-value procedures. Without further restrictions, \cite{vovk2021values} suggest that the only valid p-values are given by inverse e-values. Thus, our e-graphical approaches are strictly more powerful than the corresponding p-graphical approaches. To make the comparison more informative, we work in the sequential testing setting, where we can derive a valid p-value that is more powerful than the inverse e-value \citep{johari2022always}. 

We present results for power in terms of how often the e-procedure improves on the p-procedures, as well as by what percentage in terms of ratio of stopping times. Both give useful information for gauging how meaningful the potential improvement of e-Holm is in the sequential setting. We examine the Holm's procedure described in Section \ref{sec:e-holm} and the graphical procedure described in Section \ref{sec:dag}, specifically for a useful structuring of hypotheses in a factorial design.

For both of our simulation settings, we will be operating in a setting of a continuous stream of Gaussian data $Y_{it}\overset{i.i.d.}{\sim}\mathcal{N}(\mu_i,1)$, with known variance 1. For arm $i$, i.i.d. data comes in at times $t=1,2,...$, with the same mean $\mu_i$. In Section \ref{sec:eholm-sim}, we will be testing each arm separately, while in Section \ref{sec:edag-sim} we will be testing differences of means of arms. In either case, the way we will construct test martingales is through the sequential probability ratio test (SPRT) with known alternative, as in \cite{johari2022always}. In particular, let $\mathcal{F}_t$ denote the $\sigma$-field generated by $\{(X_{1j}, \ldots, X_{nj}): j\le t\}$. For any given predictable sequence $\hat{\mu}_s\in \mathcal{F}_{s-1}$, let, 
\begin{align}\label{eq:test_martingale} e_{it}=\prod_{s=1}^t\frac{\phi(Y_{is}-\hat{\mu}_i)}{\phi(Y_{is})}=\prod_{s=1}^t\exp(\hat{\mu}_i Y_{is}-0.5\hat{\mu}_i^2),
\end{align}

where $\phi$ is the pdf (likelihood) of the Gaussian distribution. Clearly, $e_{it}$ is a martingale with $\mathbb{E}[e_{it}] = 1$. By optional stopping theorem, for any stopping time $\tau$ with respect to the filtration $\mathcal{F}_t$, $\mathbb{E}[e_{i\tau}] \le 1$. As a result, $e_{i\tau}$ is a valid e-value and, hence, $1/e_{i\tau}$ is a valid p-value. By Ville's inequality \eqref{eq:villes} \citep{ville1939etude}, $1/\max_{t\le\tau}e_{it}$ is also a valid p-value. \cite{johari2022always} call it an always-valid p-value. We will perform e-graphical approaches on $e_{i\tau}$ and p-graphical approaches on $1/e_{i\tau}$ and $1/\max_{t\le\tau}e_{it}$, where $\tau$ is an appropriately defined stopping time, adapted to $\mathcal{F}_t$. 

The multiple testing procedures, when applied to the e-processes, provide always-valid FWER control, similar to reasoning from \cite{johari2022always}. Because $\{e_{it}\}_{t\geq1}$ is a valid martingale adapted to $\mathcal{F}_t$ for all $i\in[n]$, the stopped processes $\{e_{i\tau}\}_{i\in[n]}$ are e-values and the closure principle will provide FWER control. Any stopping criterion that uses the adjusted e-values from current or previous times $t$ is a function of the processes, and so is still adapted to the full filtration $\mathcal{F}_t$. However, if some processes are adapted to sub-filtrations, the validity may be an issue since the stopped processes may not all be e-values. \cite{choe2024combining} discuss examples of encountering and correcting for this issue.

\subsection{e-Holm simulations}\label{sec:eholm-sim}

Our setting is 20 hypotheses $H_i:\mu_i=0$, for i.i.d. data $Y_{it}\sim\mathcal{N}(\mu_i,1)$, where the known alternative for each is $\mu_{alt}$, which we vary on $\{0.5, 1, 1.5, 2\}$. For our simulation, there are 5 true alternatives. 
We define the test martingale by \eqref{eq:test_martingale} with $\hat{\mu}_i = \mu_{alt}$. 

\vspace{0.2cm}
\noindent \textbf{Time to first rejection. } First, we compare  the stopping time at which at least one hypothesis can be rejected, after accounting for multiplicity (via each procedure). 
In terms of adjusted e-values $\{e_{it}^*\}_{i=1}^n$, this stopping time is $\tau_e=\min\{t:\max_{i\in[n]}e_{it}^*\geq1/\alpha\}$, and we define $\tau_{ep}$ and $\tau_p$ analogously for the 1/e p-values and always-valid p-values.

For each value of $\mu_{alt}$, we run 1000 Monte Carlo simulations, where we save the stopping times $\tau_e,\tau_{ep},\tau_p$, all computed with the same sequence of simulated data. For practical purposes there is a maximum number of iterations of 2000, though no run reaches this threshold. Because we stop once the most significant hypothesis is rejected, its test martingale is at its maximum, and thus $\tau_e\leq \tau_p$ always, and $\tau_p=\tau_{ep}$. This follows from the discussion of the comparison between e-Bonferroni and p-Bonferroni in Section \ref{sec:eval_intro}. Thus, we only compare $\tau_e$ and $\tau_p$, and compute for $m=1000$ iterations and empirical distribution $\mathbb{P}_m$:

\begin{itemize}
    \item Comparison (probability of improving on p-Holm): $\mathbb{P}_m(\tau_e < \tau_p)$
    \item Ratio (conditional on there being a difference): $\mathbb{E}_m[\tau_e/\tau_p\mid\tau_e\neq \tau_p]$
\end{itemize}

For the former, shown in the left panel of Figure \ref{fig:holm-sim}, the e-Holm comparison is always non-negative, the greater the better. For the latter, shown in the right panel of Figure \ref{fig:holm-sim}, the e-Holm ratio is always no greater than 1, the lower the better. Both metrics are necessary to understand the efficiency of e-Holm, as a small ratio is meaningless with a small comparison.

The probability comparison in Figure \ref{fig:holm-sim} shows that the biggest improvement in likelihood of beating p-Holm comes for smaller $\mu_{alt}$, which makes intuitive sense as the stopping times are larger, leaving more margin for improvement. This reasoning also applies to the results in the right panel of Figure \ref{fig:holm-sim}, which show that the lower the stopping time the lower the potential the ratio has to be, as low as 0.6 on average for $\tau_e\neq \tau_p$ when $\mu_{alt}=2$. This intuitive rationalization aside, it is promising that all offer improvements of at least $0.05$ in terms of $\BP_m(\tau_e < \tau_p)$ and about $10\%$ for $\BE_m[\tau_e/\tau_p\mid\tau_e\neq \tau_p]$.

\vspace{0.2cm}
\noindent \textbf{Fixed-sample FWER and power. } Next, we compare the FWER and power for e-Holm after $25$ observations per hypothesis in Figure \ref{fig:holm-sim-fwer-power}. We also compare with em-Holm, which is e-Holm applied to the e-process maxima $\tilde e_{it}=\max_{s\leq t}e_{is}$ which are pseudo e-values, and em($\alpha'$)-Holm, which is adjusted such that $B_K(\alpha')=\alpha$ in the notation of Section \ref{sec:pareto}. In Figure \ref{fig:holm-sim-relative-power} we provide a zoomed in version of the power plot for clarity. As suggested by the theory, e-Holm outperforms ep-Holm and em-Holm outperforms p-Holm. Interestingly, em-Holm with the $\alpha'$ correction has similar power to p-Holm.

\begin{figure}[!tbp]
  \centering
  
{\includegraphics[width=0.48\textwidth]{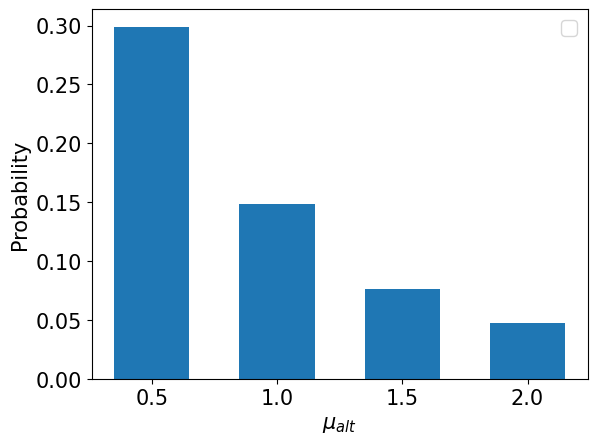}
  \hfill
 {\includegraphics[width=0.48\textwidth]{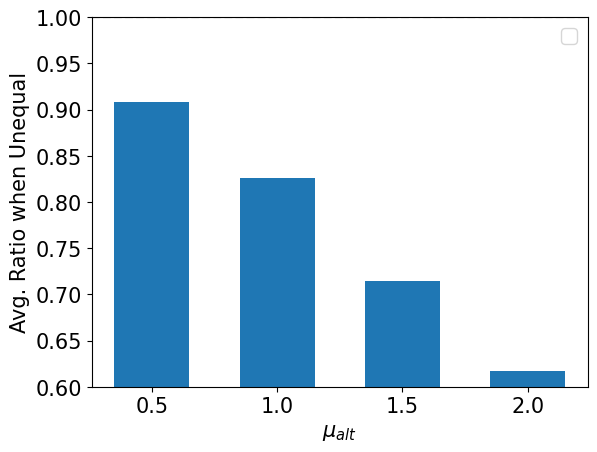}\label{holm_ratio}}
  \caption{Left: Empirical comparison from $m=1000$ of the probability of e-Holm improving on p-Holm. Right: Empirical ratio of stopping times, when they differ, between e-Holm and p-Holm using always-valid p-values (blue). Both comparisons are for a stopping time of the first rejection out of 20 hypotheses.}\label{fig:holm-sim}}
\end{figure}

\begin{figure}[!tbp]
  \centering
  
{\includegraphics[width=0.48\textwidth]{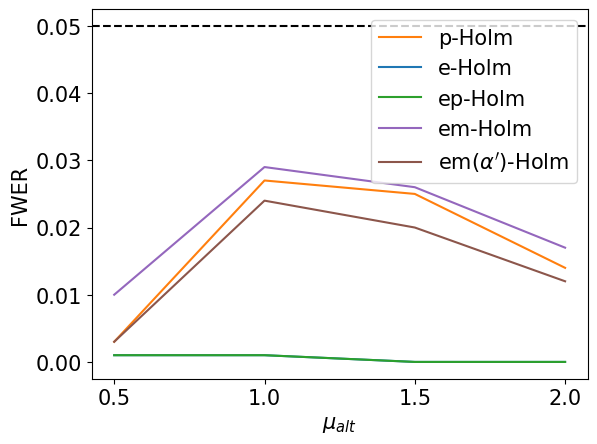}
  \hfill
 {\includegraphics[width=0.48\textwidth]{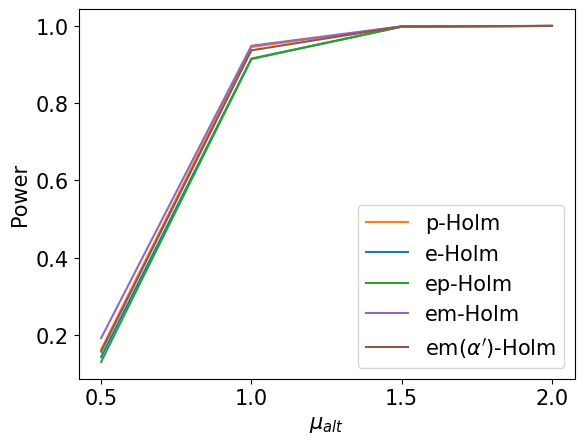}}
  \caption{FWER (left) and power (right) from $m=1000$ iterations of the p-Holm, e-Holm, ep-Holm, em-Holm, and adjusted em($\alpha'$)-Holm procedures at a fixed time of 25 steps.}\label{fig:holm-sim-fwer-power}}
\end{figure}

\begin{figure}
    \centering
    \includegraphics[width=0.5\linewidth]{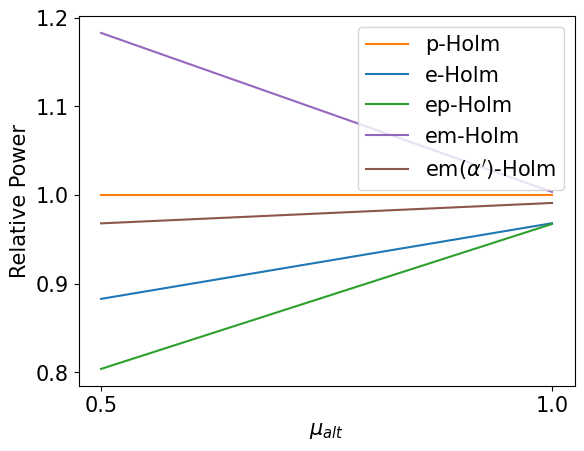}
    \caption{A relative version of the power plot from Figure \ref{fig:holm-sim-fwer-power} for only values $\mu_{alt}\in\{0.5,1\}$, where p-Holm is the baseline. This plot is included for clarity since the values are close together. The powers for all methods at $\mu_{alt}\in\{1.5,2\}$ are essentially 1.}    \label{fig:holm-sim-relative-power}
\end{figure}

\subsection{e-DAG simulations}\label{sec:edag-sim}

We continue our simulation study with a similar framework under a different model, this time with a structure to our hypotheses in comparison to the agnostic treatment of hypotheses in Section \ref{sec:eholm-sim}. We use the example of a factorial design, which is a useful framework for testing not just primary nodes, but also secondary, tertiary, etc. nodes. For example, one might be interested in the effect of a drug and the effect of gender on outcomes as primary node, but also the interaction effect of the drug for a particular gender, a secondary node.

To illustrate this framework, consider a simple three-factor model, say for A/B testing in a tech application, with three proposed treatments. The model is that data $Y$ is drawn with treatments $X\in\{0,1\}^3$ according to
\begin{align}\label{eq:three_factor}
Y=\beta_0&+\beta_1X_{1}+\beta_2X_{2}+\beta_3X_{3}\nonumber\\
&+\beta_{12}X_{1}X_{2}+\beta_{23}X_{2}X_{3}+\beta_{13}X_{1}X_{3}\\&+\beta_{123}X_{1}X_{2}X_{3}+\epsilon,\nonumber
\end{align}
where $\epsilon\sim\mathcal{N}(0,\sigma^2)$ is an i.i.d. Gaussian error term. In practice one might use a more complicated model but for simplicity we take this linear model with fixed effects and i.i.d. errors with known variance (which we take to be 1). Here the hypotheses $H_j:\beta_j=0$ correspond to primary nodes, $H_{jk}:\beta_{jk}=0$ correspond to secondary interaction nulls, and $H_{123}:\beta_{123}=0$ corresponds to the tertiary interaction null. 

In an experiment, the practitioner likely cares mostly about the primary hypotheses, but would like to test the secondary (and higher order) hypotheses as well. An approach is to use the graphical procedure to define a graph in which each primary hypothesis is a root, whose children are the secondary hypotheses containing the primary hypothesis. The children of $H_1$ are $H_{12}$ and $H_{13}$ in our three-way example, which is visualized in Figure \ref{factorial_image}.

\begin{figure}
    \centering
    \includegraphics[scale = 0.225]{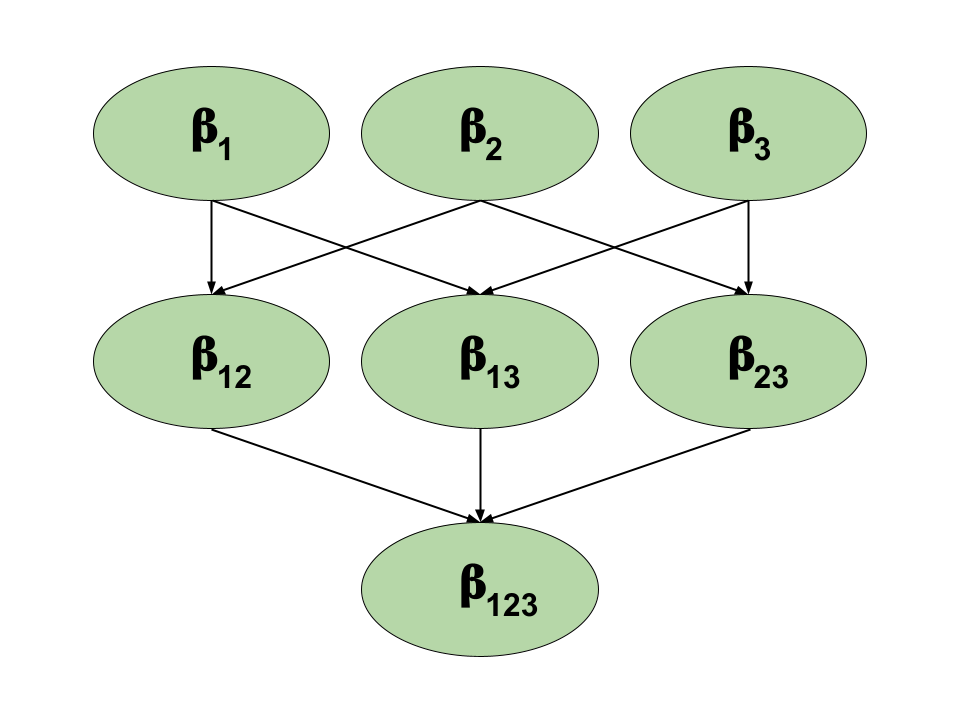}
    \caption{The transition graph for the graphical approach in a factorial design with three primary nodes.}
    \label{factorial_image}
\end{figure}

In the context of the graphical approach, it is most common to take the transition probabilities to be uniform, and the starting $\alpha$ budget to be $\alpha/p$ for each of the $p$ primary endpoints, in our case $\alpha/3$ for each of $H_1,H_2,H_3$. The result of this choice is that, in terms of rejecting any primary hypothesis, we use Bonferroni to determine if any node can be rejected. Once one of these nodes can be rejected, its budget transfers to the corresponding secondary hypotheses, which might be rejected, though they require a higher level of significance. The ancestor graph of $H_i$ is just $H_i$, so we can only reject $H_i$ if its standalone test martingale can reject at level $\alpha/3$. 

The upside is that rejecting, for example, $H_{13}$, becomes potentially easier because we can use e-Bonferroni as the local test, as if $H_1$ is rejected, we now can boost the adjusted e-value for $H_{13}$ via the value of $H_3$, if it remains unrejected. However, it is unclear a priori whether this, or the fact that the always-valid p-value gets to use the max value of the test martingale, will win out in various settings. 

In order to investigate this trade-off, we perform a similar simulation study to Section \ref{sec:eholm-sim}, with our three-way factorial design model. It is the same in that we have $m=1000$ iterations, and we vary $\beta_{13}\in\{0.5,1,1.5,2\}$. Our stopping time is the first rejection of an interaction hypothesis, where $\beta_1,\beta_3,\beta_{13}$ are the only nonzero parameters. We fix $\beta_1=\beta_3=0.5$, with this alternative being known, and $\beta_{13}$ being unknown. 

We consider a process in which at time $t$, we receive $2^3=8$ measurements corresponding to draws of our model with each of $X\in\{0,1\}^3$, which we denote with binary indices. In particular, for a given time $t$, we observe $\{Y_{tijk}\}_{i,j,k\in\{0,1\}}$, which follows the distribution \eqref{eq:three_factor} according to $Y\mid X=(i,j,k)$. For example, $Y_{t101}\sim\beta_1+\beta_3+\beta_{13}+\epsilon_{t101}$ where each error $\epsilon$ is assumed to be independent.

We can then isolate each coefficient (plus mean zero noise) from the data by taking combinations of our time $t$ data. For example, $Y_{t100}-Y_{t000}\sim\mathcal{N}(\beta_1,2)$, and similar for the other primary endpoints. Then to isolate $\beta_{13}$, we can take

\begin{align}
    Y_{t101}-Y_{t100}-Y_{t001}+Y_{t000}\sim \mathcal{N}(\beta_{13},4).
\end{align}

Because we are adding and subtracting more data points, the variance increases. This procedure generalizes via a pattern that follows the inclusion-exclusion principle. This estimator is exactly the result of an ordinary least squares on an $8\times 8$ covariate matrix with each interaction term. As such, it is the maximum likelihood estimate of the model with homoskedastic, known variance, and is optimal in the UMVUE sense. Then for each of these test statistics, we can construct a test martingale via a SPRT for an alternative mean, which for us is known for the primary hypotheses but not the interaction hypotheses $H_{12},H_{13},H_{23}$, for which we predictably estimate $\hat{\mu}_t$ based on the sample average before time $t$.

\vspace{0.2cm}
\noindent \textbf{Time to first rejection. } With these test martingales defined, we define our stopping time to be the time that some interaction hypothesis is rejected, after accounting for multiplicity via e-DAG/ep-DAG/p-DAG, i.e. $\tau_e=\min\{t:\max(e_t^{12*},e_t^{13*},e_t^{23*})\geq1/\alpha\}$, for $e_t^{13*}$ the adjusted e-value for $H_{13}$, etc., and analogously defined for $\tau_p$ and $\tau_{ep}$. In our simulations, we take $\beta_1=\beta_3=0.5$ and $\beta_{13}>0$ varies, with all other coefficients equal to 0.

As in Section \ref{sec:eholm-sim}, we record the probability of $\tau_e$ improving on $\tau_p$, but now it could be worse, and $\tau_p,\tau_{ep}$ could differ, so we record three comparisons: $\BP_m(\tau_e<\tau_p),\BP_m(\tau_e>\tau_p)$, and $\BP_m(\tau_e<\tau_{ep})$. Since e-values have other documented pros besides our procedure as outlined in Section \ref{sec:eval_intro}, and demonstrated in our e-value-based closed testing results in Section \ref{sec:e-value-closed-testing}, knowing the latter is relevant. We also record the ratio conditional on a difference, both for $\tau_p$ and $\tau_{ep}$, because they might differ. In notation, these quantities are $\BE_m[\tau_e/\tau_p\mid\tau_e\neq \tau_p]$ and $\BE_m[\tau_e/\tau_{ep}\mid\tau_e\neq \tau_{ep}]$.

We do two simulations. The first is for primary budget, meaning each of $H_1,H_2,H_3$ is assigned budget $\alpha/3$, and no budget for the secondary and tertiary hypotheses. We also perform the same analysis for equal budget among all seven hypotheses, which allows $H_{13}$, the subject of the simulation study, to be rejected even if $H_1$ and $H_3$ are both unrejected. This might be desirable if the secondary hypothesis is of interest separately from the primary ones. The probability comparisons are given in Figures \ref{fig:DAG_prob-p} and \ref{fig:DAG_prob-ep} and the stopping time ratios are given in Figure \ref{fig:DAG_ratio}.

In general, rejection tends to be faster for e-DAG for higher values of $\beta_{13}$, with the only favorable ratio occurring for $\beta_{13}=2$. However, the results show that the gains over ep-DAG are consistent and substantial, with settings in which e-DAG is able to improve upon p-DAG.
Just like the simulations show, the plots with equal budget experiment generally follow similar trends to the plots for primary budget, but with improvements across the board. Notably, the probability of p-DAG beating e-DAG becomes very low for $\beta_{13}=2$, with a corresponding ratio of less than 0.95. Both simulation settings (primary and equal budget for factorial design) show that e-DAG can consistently be a notable improvement on ep-DAG, while having mixed results compared to p-DAG, albeit with many promising settings.

\vspace{0.2cm}
\noindent \textbf{Fixed-sample FWER and power. }We also show the power in a fixed sample size of $50$ steps in Figure \ref{fig:DAG_power}. 

In this case the coefficients that are true alternatives are $\beta_1,\beta_3,\beta_{13}$. Because the only parameter changing between instances of the simulation is $\beta_{13}$, the p-DAG FWER remains the same for all coefficients, at $0.006$ for primary budget and at $0.006$ for equal budget, while the FWER for e-DAG and ep-DAG is $0.002$ for primary budget and $0.001$ for equal budget. 

\begin{figure}[!tbp]
  \centering
  
{\includegraphics[width=0.48\textwidth]{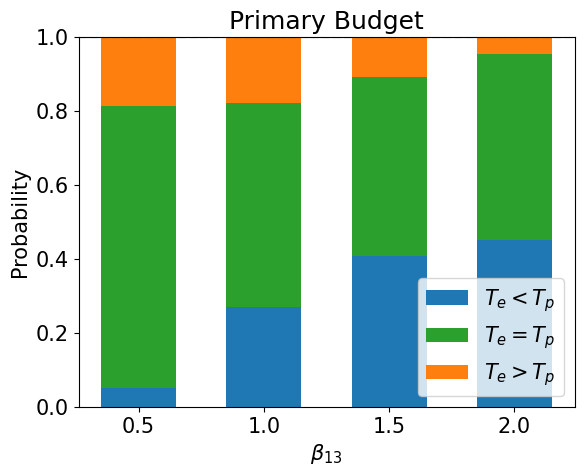}\label{PB_prob-p}}
  \hfill
 {\includegraphics[width=0.48\textwidth]{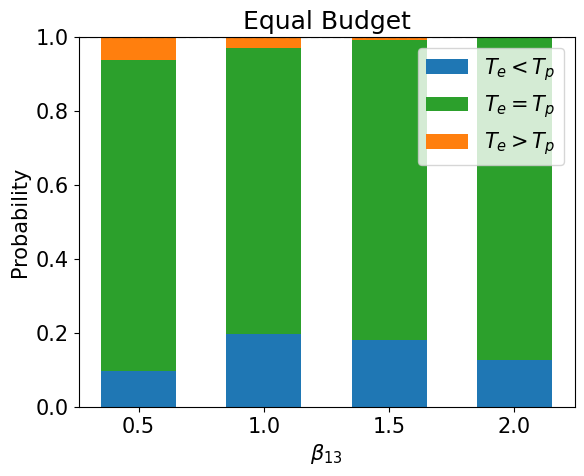}\label{EB_prob-p}
  \caption{Empirical probability from $m=1000$ iterations of the e-DAG stopping time being less than (blue), equal to (green), or greater than (orange) p-DAG. The stopping time is when we reject any secondary hypothesis in a factorial design graphical model. With primary budget and equal budget respectively.}\label{fig:DAG_prob-p}}
\end{figure}

\begin{figure}[!tbp]
  \centering
  
{\includegraphics[width=0.48\textwidth]{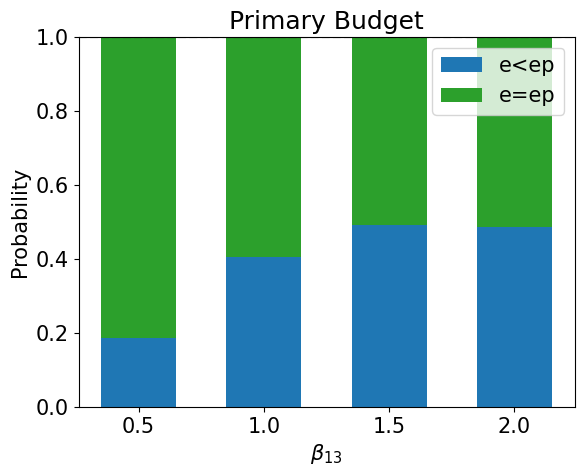}\label{PB_prob-ep}}
  \hfill
 {\includegraphics[width=0.48\textwidth]{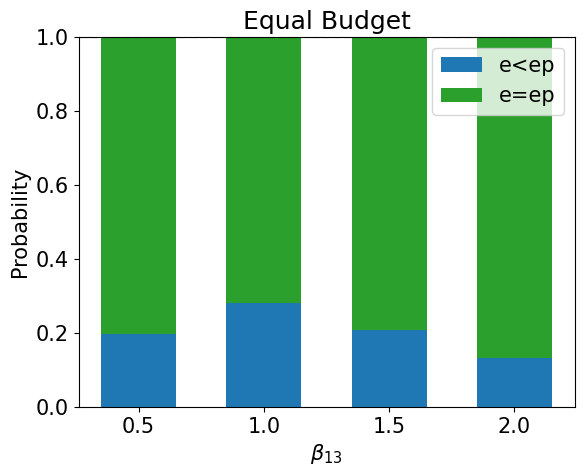}\label{EB_prob-ep}
  \caption{Empirical probability from $m=1000$ iterations of the e-DAG stopping time being less than (blue), equal to (green) ep-DAG. ep-DAG can never reject faster than e-DAG by construction. The stopping time is when we reject any secondary hypothesis in a factorial design graphical model. With primary budget and equal budget respectively.}\label{fig:DAG_prob-ep}}
\end{figure}

\begin{figure}[!tbp]
  \centering
  
{\includegraphics[width=0.48\textwidth]{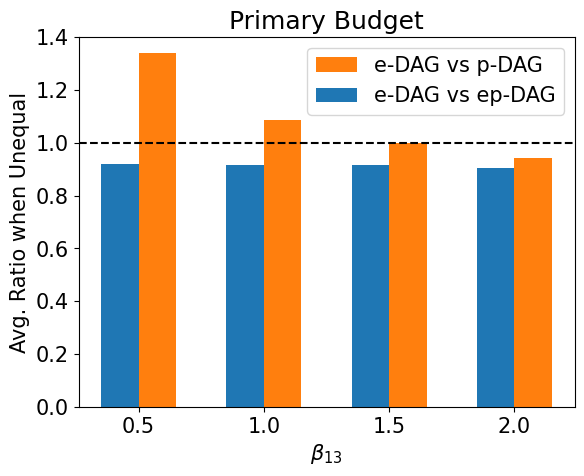}\label{PB_ratio}}
  \hfill
 {\includegraphics[width=0.48\textwidth]{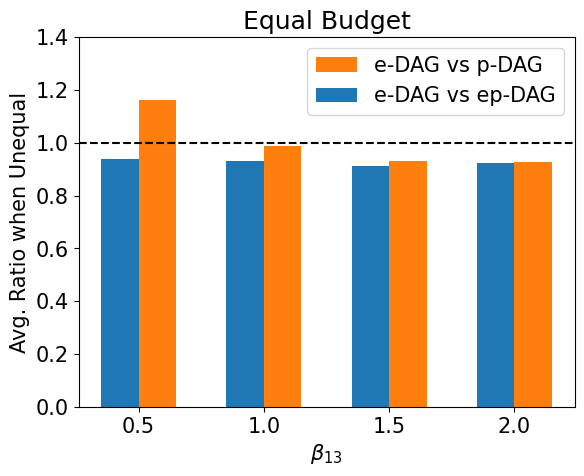}\label{EB_ratio}}
  \caption{Empirical ratio from $m=1000$ iterations of the stopping times, when they differ, from Figures \ref{fig:DAG_prob-p} and \ref{fig:DAG_prob-ep} for e-DAG against p-DAG (orange) and ep-DAG (blue). With primary budget and equal budget respectively.}\label{fig:DAG_ratio}
\end{figure}

\begin{figure}[!tbp]
  \centering
  
{\includegraphics[width=0.48\textwidth]{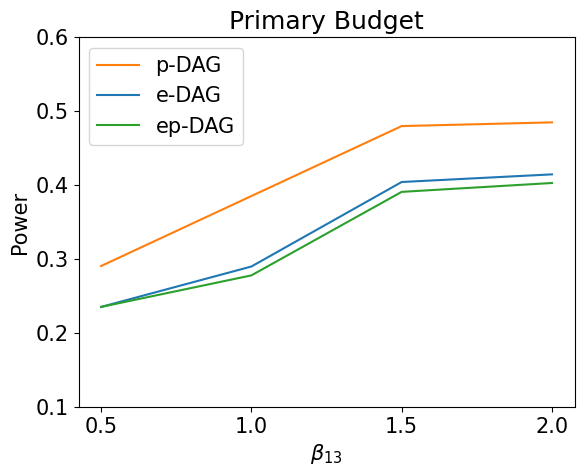}\label{PB_power}}
  \hfill
 {\includegraphics[width=0.48\textwidth]{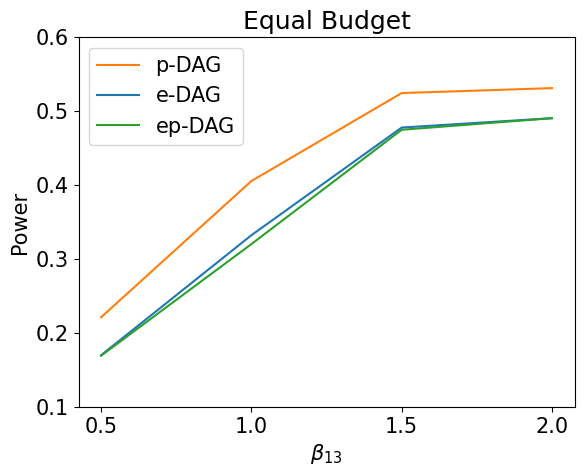}\label{EB_power}
  \caption{Power from $m=1000$ iterations of the p-DAG, e-DAG, and ep-DAG procedures, after a fixed 50 steps. The only change between the simulations is the alternative hypothesis, which is why the functions level out. For primary budget, the p-DAG FWER is 0.006, and for equal budget, the p-DAG FWER is 0.006. For both e-DAG and ep-DAG the FWER is 0.002 for primary budget and 0.001 for equal budget.}\label{fig:DAG_power}}
\end{figure}

\section{Semi-Synthetic Simulations with New Yorker Caption Contest Data}\label{sec:new-yorker}

\subsection{New Yorker Caption Contest Data}

To illustrate the utility of e-closed testing in a real data setting where the sequential data has a natural dependence, we apply both e-Holm and e-DAG to testing problems related to the New Yorker Caption Contest data collected by \cite{zhang2024humor}. The data consist of aggregated vote totals and corresponding average ratings assigned to various captions submitted for each cartoon. The data was collected from sequential user votes, where the relevant caption was determined by a bandit sampling algorithm. This dependent sequential setting provides a popular dataset for sequential testing, and has been used as an example by such papers as \cite{yang2017}.

Since the dataset is aggregated for each caption, following previous work, we design a semi-synthetic simulation to mimic the multi-armed bandits setting, with the empirical averages treated as the ground truth. Specifically, we consider contest \#895, which received 4958 captions with 777820 total votes cast. Users rated captions on a scale of $\{1,2,3\}$ and the average rating for the top-rated caption was 1.6323, and 64 captions had a mean above 1.5. We have access to both the average rating $\mu_i$ for each caption, as well as a vector of rating probabilities $(p_i(1),p_i(2),p_i(3))$.

Using a lilUCB sampling algorithm \citep{jamieson2014lil} detailed in Appendix \ref{sec:simulation-appendix}, which selects a random arm with probability $\epsilon=0.4$ in our simulation, and selects an arm with the highest upper confidence bound (UCB) otherwise, we simulate votes   for a subset of captions. Using this framework, we present two semi-synthetic simulations. First, we test the null hypothesis $H_{0,i}:\mu_i\leq1.5$, i.e. whether caption $i$ has a mean below 1.5 and use Holm's procedure to control the FWER. Second, we define three semantic features describing the content of each caption, and test the main effects and interaction effects for these 6 average contrasts, controlling the FWER with e-DAG as in Section \ref{sec:edag-sim}.

\subsection{e-Holm for Caption-testing simulation} We compute an e-process $e_{it}$ for each caption $i$ whose multiplier is 1 if it is not chosen at pull $s$ and a betting score otherwise, where $k_s$ is the arm that is chosen at time $s$:
\begin{align*}
e_{it}=\prod_{s=1}^t\Big[\mathbf{1}\{k_s\neq i\}+\mathbf{1}\{k_s=i\}\exp(\eta(Y_s-\mu_0)-\psi)\Big]
\end{align*}
where $Y_s$ is the value of the arm at pull $s$, $\eta=0.25$ is a tuning parameter and $\mu_0$ is the null against which we are testing. The parameter $\psi$ is a normalization parameter that ensures the betting score has conditional expectation $\leq1$ and that the process is an e-process for each $i$.

We apply Holm's procedure to $n_1=20$ non-nulls and $n_0=30$ nulls, and track the time to first rejection, in the same manner as in Section \ref{sec:eholm-sim}. We compute the average time until the first Holm-adjusted rejection for e-Holm and p-Holm, in terms of the overall number of pulls $T_0$ and the number of pulls of the best arm $T_{0,i^*}$ where $i^*$ is the best arm. We know that e-Holm rejects at least as fast, and in this case the average ratio of $T_0$ is 0.988, and the average ratio of $T_{0,i^*}$ is 0.981. In terms of the probability of giving an improvement, e-Holm improves on p-Holm with probability 0.79 for $T_0$ and with probability 0.78 for $T_{0,i^*}$.

\subsection{e-DAG for semantic feature-testing simulation} We assign each caption to one of eight groups according to its factorial assignment in $\{0,1\}^3$. With the same contrasts as in Section \ref{sec:edag-sim}, we can define a true effect as the corresponding contrast between the group-level average true ratings. For example, we define $\beta_{13}$ for this population to be
\begin{align*}
    \beta_{13}=\frac{1}{n_{101}}\sum_{a_i=101}\mu_i-\frac{1}{n_{001}}\sum_{a_i=001}\mu_i-\frac{1}{n_{100}}\sum_{a_i=100}\mu_i+\frac{1}{n_{000}}\sum_{a_i=000}\mu_i
\end{align*}
where $n_a$ is the number of captions with a given group, $a_i$ is the group for caption $i$, and $\mu_i$ is the average rating for caption $i$. In the case where the design is balanced and $n_a$ is constant over $a$, this is exactly the result of an ANOVA linear model for the caption data. We model this framework by randomly sampling $m=6$ captions from each group to get a total of $n=8m=48$ captions. 

We run the same bandit sampler as in the Holm simulation, with $T=1000$ pulls, and test the main and interaction effects via e-processes. To do this we employ an augmented inverse propensity weighting (AIPW) estimator, which we detail further in Section \ref{sec:simulation-appendix}. The stopping times for the main effects will be identical, so we track when the first interaction effect is rejected after adjustment, given by $T_{int}$.

We find that the average ratio of $T_{int}$ between e-DAG and p-DAG is 0.897, indicating that e-DAG rejects faster, on average. Similarly, the probability that $T_{int}$ improves the stopping time is 0.93.

\subsection{E-Holm for fixed-sample-size simulation}\label{sec:fixed-sample} Lastly, we present a semi-synthetic simulation that is separate from the bandit algorithm structure, to demonstrate the difference in these methods for a fixed sample setting with traditional p-values. We use the same formulation as our simulation for Holm, selecting $n_1=20$ alternative hypotheses and $n_0=30$ null hypotheses. The $n_1$ alternatives are the largest means, and the $n_0$ nulls are the largest means with $\mu\leq1.5$, the closest nulls to the boundary. This design means that there is a sizable gap between the alternative means and the boundary.

For each of $n=50$ captions, we draw $m$ observations, and compute a p-value from a one-sided one sample t-test. We vary $m$ to compare FWER and power across sample sizes. Typically, a fixed sample size e-value is much less powerful than the corresponding optimal p-value for each individual hypothesis so the power loss will be amplified for multiple testing. Nevertheless, Theorem \ref{th:fwer-pareto-bound} allows us to use e-Holm on the inverse p-values that are pseudo e-values, which we call em-Holm (to align with em-Holm from the preceding subsections), giving approximate FWER control. Because we now work with a fixed sample outside of the bandit contextstock, we have independence between e-values and can apply em-Holm. We can also compute $\alpha'=\sup\{\alpha^*\in(0,1):B_n(\alpha^*)\leq\alpha\}$ and apply e-Holm to pseudo e-values at level $\alpha'$ to guarantee FWER control. For our simulation, $\alpha=0.05$ and $\alpha'=0.038$. In our semi-synthetic simulation, we compare these three methods: p-Holm, em-Holm, and em($\alpha'$)-Holm.

\begin{figure}
    \centering
    {\includegraphics[width=0.48\textwidth]{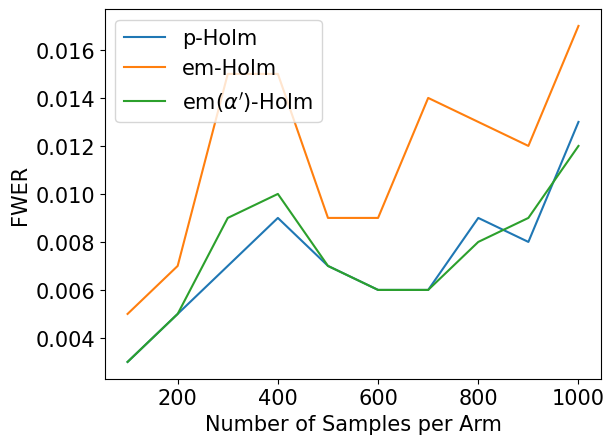}}
  \hfill
 {\includegraphics[width=0.48\textwidth]{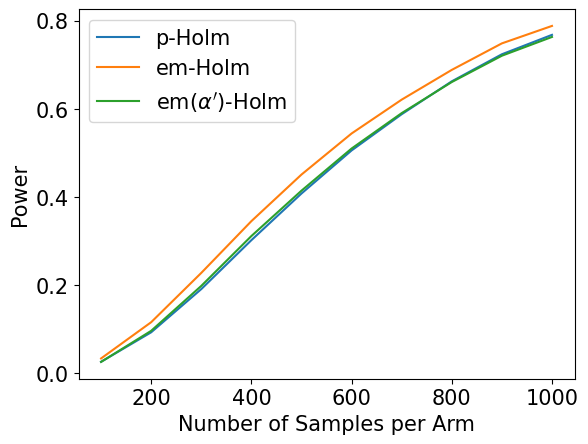}}
    \caption{FWER and power for testing $n=50$ caption means with $n_1=20$ alternatives at level $\alpha=0.05$. We vary the number of observations per caption. The implied FWER for em-Holm is $B_{30}(0.05)=0.07$ but it still controls FWER empirically.}
    \label{fig:fixed-sample-sims}
\end{figure}

Figure \ref{fig:fixed-sample-sims} shows that in this case, the pseudo e-value approach is more powerful than p-Holm while maintaining FWER control. Even when applying the adjusted level $\alpha'=B_{50}^{-1}(\alpha)=.038$, the e-Holm procedure on pseudo e-values is similarly powerful to p-Holm, while guaranteeing FWER control. This adjusted level accounts for the possibility that all hypotheses are nulls, so it is conservative since there are actually 30 nulls.

\section{Conclusion and discussion}
In this paper, we formalized the framework of e-value-based closed testing, in which each local test in the closure principle for FWER control is based on an e-value. We showed that e-value-based closed testing provides the ability to select the level of control $\alpha$ post-hoc, both in the traditional and anytime-valid settings. We also show  e-value-based closed testing applied to any global e-processes provides always-valid FWER control for fixed $\alpha$ and that applied to the running maxima and calibrated with an adjuster achieves both post-hoc and always validity. Furthermore, when $\alpha\rightarrow 0$, the closure of weighted e-Bonferroni local tests applied to pseudo-e-values as inverse running maxima of e-processes provides strong FWER control. 

A large class of closed tests are based on weighted p-Bonferroni based on a graph structure. We extend this graphical approach to weighted e-Bonferroni by developing efficient algorithms to compute the rejection set for adjusted e-values. For the procedure based on unweighted Bonferroni, we build on the e-Holm algorithm of \cite{vovk2021values,vovk2023confidence} to get a simple rejection threshold, and quantify the power increase compared to p-Holm with inverse-e p-values. For the graphical approach on directed acyclic graphs, we develop an efficient algorithm to compute the e-DAG closed test, and refine the algorithm to an $O(n)$ algorithm for the e-Fallback procedure which is based on a chain graph.

\subsection{Improved thresholds for local tests}
In earlier sections, we focused on the standard e-value test that uses a threshold of $1/\alpha$ for each local hypothesis. Recent studies have demonstrated that this threshold can be uniformly improved. Specifically, \cite{ramdas2026randomized} and \cite{koning2024continuous} show that $1/\alpha$ can be replaced by $U / \alpha$ where $U$ is uniformly distributed on $[0, 1]$ and its randomness is exogenous. For closed testing, we can define the local test as $e_{It}\ge U_{I} / \alpha$ for a collection of uniform random variables $\{U_I: I \subset [n]\}$ that are independent of the e-values. By Lemma \ref{lem:key}, the FWER is bounded by 
\[\mathbb{P}\left(e_{\mathcal{H}_0}\ge \frac{U_{\mathcal{H}_0}}{\alpha}\right) \le \BE[\min(1,\alpha e_{\mathcal{H}_0})]\le \BE[\alpha e_{\mathcal{H}_0}] \le \alpha.\]
As a result, the randomized e-value-based closed testing controls FWER for any arbitrary dependence structure among the uniforms. 
In the sequential setting, if we use a time-invariant uniform random variable for each local test, we can similarly prove the anytime-valid FWER control. However, it remains unclear whether other properties discussed in Section \ref{sec:e-value-closed-testing}, such as the post-hoc or always-valid FWER control, continue to hold. 

With additional distributional assumptions on the e-values, \cite{blier2024improved} show that the threshold could be improved in a deterministic fashion. For example, if the e-value has a decreasing density, its $(1-\alpha)$-th quantile is bounded by $1/(2\alpha)$ (Theorem 2 (i) of \cite{blier2024improved}). Applying this result to our setting, if $e_1, \ldots, e_n$ for elementary hypotheses are independent with decreasing densities, any weighted average of them still has a decreasing density. As a consequence, the closure of any collection of weighted e-Bonferroni test at level $2\alpha$ controls FWER at $\alpha$. 

In the case where test statistics have a well-defined correlation structure under the null, \cite{bretz2011graphical} describe a procedure that adjusts the threshold by a constant $c_I$ for each subset $I$ such that the local test is still valid. In principle, the same technique could be employed for e-values with a known covariance structure, changing the rejection threshold for $H_I$ to $1/(c_I\alpha)$ for $c_I\geq1$, even if $c_Ie_I$ is not a valid e-value, as long as the resulting local test for $H_I$ is valid. We leave the derivation of such constants for future work.

\subsection{FWER control under logical constraints}
In previous sections, we focus on the case with logically independent hypotheses -- that is, any combination of true and false nulls is considered possible. However, in many real-world applications, the hypotheses are logically related, meaning that some combinations of true/false status are ruled out a priori \citep{shaffer1986modified, goeman2010sequential,loper2022smoothed}. Logical constraints \citep{shaffer1986modified} are deterministic rules that restrict the allowable configurations of null and alternative hypotheses. For example, under a multi-factor model like \eqref{eq:three_factor}, the hereditary principle is often invoked as a structural prior that interactions are unlikely to be present unless supported by main effects. More formally, the strong hereditary principle states that $H_{ij}$ must be null if $H_i$ or $H_j$ is null, and the weak hereditary principle states that $H_{ij}$ must be null if $H_i$ and $H_j$ are nulls. 

Under the logical constraints, the intersection hypothesis $H_I$ is equivalent to $H_{J(I)}$ where $J(I)\subset I$ is the minimal subset of $I$ such that no hypothesis in $J(I)$ is implied by any hypothesis in $I$. For example, under the strong hereditary principle, $J(I) = \{\{1\}, \{2\}\}$ when $I = \{\{1\}, \{2\}, \{1,2\}, \{1,3\}, \{1,2,3\}\}$. The logical constraints hence lead to a reduction of the number of hypotheses -- the adjusted e-value can be defined as $e_i^* = \min_{I\ni i}e_{J(I)}$, which is at least as large as the adjusted e-value in the absence of logical constraints. All theoretical properties discussed in Section \ref{sec:e-value-closed-testing} remain valid for this procedure. Nonetheless, the dynamic programming-based algorithm for e-graphical approaches needs modification. We leave this for future research. 

\subsection{Post-hoc false discovery proportion inference with e-values}
The influential work by \cite{goeman2011multiple} proposes a generic approach to construct a simultaneous confidence bound on the false discovery proportion (FDP) based on closed testing. Specifically, for each given subset $\mathcal{R}\subset [n]$, let 
\[\ell_{1-\alpha}(\mathcal{R}) = \min\{|\mathcal{R}\setminus J|: \phi_I = 0\,\, \forall I\supset J\}.\]
Then $\ell_{1-\alpha}(\mathcal{R})$ is a uniform lower confidence bound on the number of true discoveries in $\mathcal{R}$:
\[\mathbb{P}\left( |\mathcal{R}\setminus \mathcal{H}_0|\ge \ell_{1-\alpha}(\mathcal{R}), \,\, \forall \mathcal{R}\subset [n]\right)\ge 1- \alpha.\]
For p-value-based closed testing, the consonance property leads to fast algorithms \citep{goeman2011multiple}. \cite{vovk2023confidence, vovk2024true} develop an efficient algorithm for the e-Holm procedure, leveraging the symmetry of the unweighted e-Bonferroni procedure. For general e-graphical approaches, we can compute $\ell_{1-\alpha}(\mathcal{R})$ by brute force when $n$ is small. For large $n$, the computation becomes more involved due to the asymmetry and we leave this for future research. 

\bibliography{reference}

@article{wasserman2020universal,
  title={Universal inference},
  author={Wasserman, Larry and Ramdas, Aaditya and Balakrishnan, Sivaraman},
  journal={Proceedings of the National Academy of Sciences},
  volume={117},
  number={29},
  pages={16880--16890},
  year={2020},
  publisher={National Acad Sciences}
}

@article{vovk2021values,
  title={E-values: Calibration, combination and applications},
  author={Vovk, Vladimir and Wang, Ruodu},
  journal={The Annals of Statistics},
  volume={49},
  number={3},
  pages={1736--1754},
  year={2021},
  publisher={Institute of Mathematical Statistics}
}

@article{gabriel1969simultaneous,
  title={Simultaneous test procedures--some theory of multiple comparisons},
  author={Gabriel, K Ruben},
  journal={The Annals of Mathematical Statistics},
  volume={40},
  number={1},
  pages={224--250},
  year={1969},
  publisher={Institute of Mathematical Statistics}
}

@article{list2019multiple,
  title={Multiple hypothesis testing in experimental economics},
  author={List, John A and Shaikh, Azeem M and Xu, Yang},
  journal={Experimental Economics},
  volume={22},
  pages={773--793},
  year={2019},
  publisher={Springer}
}

@article{viviano2024model,
  title={A model of multiple hypothesis testing},
  author={Viviano, Davide and Wuthrich, Kaspar and Niehaus, Paul},
  journal={arXiv preprint arXiv:2104.13367},
  year={2024}
}

@article{romano2005stepwise,
  title={Stepwise multiple testing as formalized data snooping},
  author={Romano, Joseph P and Wolf, Michael},
  journal={Econometrica},
  volume={73},
  number={4},
  pages={1237--1282},
  year={2005},
  publisher={Wiley Online Library}
}

@article{vickerstaff2019methods,
  title={Methods to adjust for multiple comparisons in the analysis and sample size calculation of randomised controlled trials with multiple primary outcomes},
  author={Vickerstaff, Victoria and Omar, Rumana Z and Ambler, Gareth},
  journal={BMC medical research methodology},
  volume={19},
  pages={1--13},
  year={2019},
  publisher={Springer}
}

@article{pocock1987analysis,
  title={The analysis of multiple endpoints in clinical trials},
  author={Pocock, Stuart J and Geller, Nancy L and Tsiatis, Anastasios A},
  journal={Biometrics},
  pages={487--498},
  year={1987},
  publisher={JSTOR}
}

@article{bretz2009graphical,
  title={A graphical approach to sequentially rejective multiple test procedures},
  author={Bretz, Frank and Maurer, Willi and Brannath, Werner and Posch, Martin},
  journal={Statistics in medicine},
  volume={28},
  number={4},
  pages={586--604},
  year={2009},
  publisher={Wiley Online Library}
}

@article{johari2022always,
  title={Always valid inference: Continuous monitoring of a/b tests},
  author={Johari, Ramesh and Koomen, Pete and Pekelis, Leonid and Walsh, David},
  journal={Operations Research},
  volume={70},
  number={3},
  pages={1806--1821},
  year={2022},
  publisher={INFORMS}
}

@article{angelopoulos2021learn,
  title={Learn then test: Calibrating predictive algorithms to achieve risk control},
  author={Angelopoulos, Anastasios N and Bates, Stephen and Cand{\`e}s, Emmanuel J and Jordan, Michael I and Lei, Lihua},
  journal={arXiv preprint arXiv:2110.01052},
  year={2021}
}

@article{tse2022note,
  title={A note on universal inference},
  author={Tse, Timmy and Davison, Anthony C},
  journal={Stat},
  volume={11},
  number={1},
  pages={e501},
  year={2022},
  publisher={Wiley Online Library}
}

@article{park2023robust,
  title={Robust Universal Inference},
  author={Park, Beomjo and Balakrishnan, Sivaraman and Wasserman, Larry},
  journal={arXiv preprint arXiv:2307.04034},
  year={2023}
}

@article{spector2023discussion,
  title={A Discussion of “A Note on Universal Inference” by Tse and Davison},
  author={Spector, Asher and Cand{\`e}s, Emmanuel and Lei, Lihua},
  journal={Stat},
  volume={12},
  number={1},
  pages={e570},
  year={2023},
  publisher={Wiley Online Library}
}

@article{storey2004strong,
  title={Strong control, conservative point estimation and simultaneous conservative consistency of false discovery rates: a unified approach},
  author={Storey, John D and Taylor, Jonathan E and Siegmund, David},
  journal={Journal of the Royal Statistical Society Series B: Statistical Methodology},
  volume={66},
  number={1},
  pages={187--205},
  year={2004},
  publisher={Oxford University Press}
}

@article{benjamini2006adaptive,
  title={Adaptive linear step-up procedures that control the false discovery rate},
  author={Benjamini, Yoav and Krieger, Abba M and Yekutieli, Daniel},
  journal={Biometrika},
  volume={93},
  number={3},
  pages={491--507},
  year={2006},
  publisher={Oxford University Press}
}

@article{grunwald2024beyond,
  title={Beyond Neyman--Pearson: E-values enable hypothesis testing with a data-driven alpha},
  author={Gr{\"u}nwald, Peter D},
  journal={Proceedings of the National Academy of Sciences},
  volume={121},
  number={39},
  pages={e2302098121},
  year={2024},
  publisher={National Academy of Sciences}
}

@article{hemerik2024choosing,
  title={Choosing alpha post hoc: the danger of multiple standard significance thresholds},
  author={Hemerik, Jesse and Koning, Nick W},
  journal={arXiv preprint arXiv:2410.02306},
  year={2024}
}

@article{wiens2005,
    author = {Wiens, Brian L. and Dmitrienko, Alexei},
    title = {The fallback procedure for evaluating a single family of hypotheses},
    journal = {Journal of Biopharmaceutical Statistics},
    year = {2005},
    volume = {15},
    number = {6},
    pages = {929-42}
}

@book{ville1939etude,
  title={Etude critique de la notion de collectif},
  author={Ville, Jean},
  year={1939},
  publisher={Gauthier-Villars Paris}
}

@article{sonnemann1988Vollständigkeitssätze,
  author = {Sonnemann, E. and Finner, H.},
  title = {Vollständigkeitssätze für multiple testprobleme},
  journal = {Medizinische Informatik Und Statistik},
  year = {1988},
  pages = {121-135}
}

@article{ramdas2023gametheoretic,
author = {Aaditya Ramdas and Peter Gr{\"u}nwald and Vladimir Vovk and Glenn Shafer},
title = {{Game-Theoretic Statistics and Safe Anytime-Valid Inference}},
volume = {38},
journal = {Statistical Science},
number = {4},
publisher = {Institute of Mathematical Statistics},
pages = {576 -- 601},
year = {2023},
doi = {10.1214/23-STS894},
URL = {https://doi.org/10.1214/23-STS894}
}

@article{waudby2024estimating,
  title={Estimating means of bounded random variables by betting},
  author={Waudby-Smith, Ian and Ramdas, Aaditya},
  journal={Journal of the Royal Statistical Society Series B: Statistical Methodology},
  volume={86},
  number={1},
  pages={1--27},
  year={2024},
  publisher={Oxford University Press US}
}

@article{wang2022false,
  title={False discovery rate control with e-values},
  author={Wang, Ruodu and Ramdas, Aaditya},
  journal={Journal of the Royal Statistical Society Series B: Statistical Methodology},
  volume={84},
  number={3},
  pages={822--852},
  year={2022},
  publisher={Oxford University Press}
}

@article{wang2024admissiblewaymergingevalues,
      title={The only admissible way of merging e-values}, 
      author={Ruodu Wang},
      year={2024},
      journal={arXiv preprint arXiv:2409.19888}
}

@article{koning2024posthocalphahypothesistesting,
      title={Post-hoc $\alpha$ Hypothesis Testing and the Post-hoc $p$-value}, 
      author={Nick W. Koning},
      journal={arXiv preprint arXiv:2312.08040},
      year={2024} 
}

@article{vovk2023confidence,
  title={Confidence and discoveries with e-values},
  author={Vovk, Vladimir and Wang, Ruodu},
  journal={Statistical Science},
  volume={38},
  number={2},
  pages={329--354},
  year={2023},
  publisher={Institute of Mathematical Statistics}
}

@article{vovk2024true,
  title={True and false discoveries with independent and sequential e-values},
  author={Vovk, Vladimir and Wang, Ruodu},
  journal={Canadian Journal of Statistics},
  volume={52},
  number={4},
  pages={e11833},
  year={2024},
  publisher={Wiley Online Library}
}

@article{goeman2011multiple,
  title={Multiple Testing for Exploratory Research},
  author={Goeman, Jelle J and Solari, Aldo},
  journal={Statistical Science},
  volume={26},
  number={4},
  pages={584--597},
  year={2011}
}

@article{fischer2024online,
  title={Online closed testing with e-values},
  author={Fischer, Lasse and Ramdas, Aaditya},
  journal={arXiv preprint arXiv:2407.15733},
  year={2024}
}

@article{tavyrikov2025carefree,
  title={Carefree multiple testing with e-processes},
  author={Tavyrikov, Yury and Goeman, Jelle J and de Heide, Rianne},
  journal={arXiv preprint arXiv:2501.19360},
  year={2025}
}

@article{koning2025sequentializing,
  title={Sequentializing a Test: Anytime Validity is Free},
  author={Koning, Nick W and van Meer, Sam},
  journal={arXiv preprint arXiv:2501.03982},
  year={2025}
}

@article{koning2024continuous,
  title={Continuous Testing},
  author={Koning, Nick W},
  journal={arXiv preprint arXiv:2409.05654},
  year={2024}
}

@article{ramdas2020admissible,
  title={Admissible anytime-valid sequential inference must rely on nonnegative martingales},
  author={Ramdas, Aaditya and Ruf, Johannes and Larsson, Martin and Koolen, Wouter},
  journal={arXiv preprint arXiv:2009.03167},
  year={2020}
}

@article{choe2024combining,
  title={Combining Evidence Across Filtrations Using Adjusters},
  author={Choe, Yo Joong and Ramdas, Aaditya},
  journal={arXiv preprint arXiv:2402.09698},
  year={2024}
}

@article{koolen2014buy,
  title={Buy low, sell high},
  author={Koolen, Wouter M and Vovk, Vladimir},
  journal={Theoretical Computer Science},
  volume={558},
  pages={144--158},
  year={2014},
  publisher={Elsevier}
}

@article{shafer2011test,
  title={Test Martingales, Bayes Factors and p-Values},
  author={Shafer, Glenn and Shen, Alexander and Vereshchagin, Nikolai and Vovk, Vladimir},
  journal={Statistical Science},
  volume={26},
  number={1},
  pages={84--101},
  year={2011}
}

@article{shafer2021testing,
    author = {Shafer, Glenn},
    title = {Testing by Betting: A Strategy for Statistical and Scientific Communication},
    journal = {Journal of the Royal Statistical Society Series A: Statistics in Society},
    volume = {184},
    number = {2},
    pages = {407-431},
    year = {2021},
}

@article{ramdas2026randomized,
  title={Randomized and exchangeable improvements of Markov’s, Chebyshev’s and Chernoff’s inequalities},
  author={Ramdas, Aaditya and Manole, Tudor},
  journal={Statistical Science},
  volume={41},
  number={1},
  pages={121--142},
  year={2026},
  publisher={Institute of Mathematical Statistics}
}

@article{blier2024improved,
  title={Improved thresholds for e-values},
  author={Blier-Wong, Christopher and Wang, Ruodu},
  journal={arXiv preprint arXiv:2408.11307},
  year={2024}
}

@article{goeman2010sequential,
  title={The sequential rejection principle of familywise error control},
  author={Goeman, Jelle J and Solari, Aldo},
  journal={The Annals of Statistics},
  pages={3782--3810},
  year={2010},
  publisher={JSTOR}
}

@article{goeman2021only,
  title={Only closed testing procedures are admissible for controlling false discovery proportions},
  author={Goeman, Jelle J and Hemerik, Jesse and Solari, Aldo},
  journal={The Annals of Statistics},
  volume={49},
  number={2},
  pages={1218--1238},
  year={2021},
  publisher={JSTOR}
}

@article{wang2025anytime,
  title={Anytime-valid FDR control with the stopped e-BH procedure},
  author={Wang, Hongjian and Dandapanthula, Sanjit and Ramdas, Aaditya},
  journal={arXiv preprint arXiv:2502.08539},
  year={2025}
}

@article{shaffer1986modified,
  title={Modified sequentially rejective multiple test procedures},
  author={Shaffer, Juliet Popper},
  journal={Journal of the American Statistical Association},
  volume={81},
  number={395},
  pages={826--831},
  year={1986},
  publisher={Taylor \& Francis}
}

@article{loper2022smoothed,
  title={Smoothed nested testing on directed acyclic graphs},
  author={Loper, JH and Lei, L and Fithian, W and Tansey, W},
  journal={Biometrika},
  volume={109},
  number={2},
  pages={457--471},
  year={2022}
}

@article{xu2026bringing,
  title={Bringing closure to false discovery rate control: A general principle for multiple testing},
  author={Xu, Ziyu and Solari, Aldo and Fischer, Lasse and de Heide, Rianne and Ramdas, Aaditya and Goeman, Jelle},
  journal={arXiv preprint arXiv:2509.02517v2},
  year={2026}
}

@article{ignatiadis2025asymptotic,
  title={Asymptotic and compound e-values: multiple testing and empirical Bayes},
  author={Ignatiadis, Nikolaos and Wang, Ruodu and Ramdas, Aaditya},
  journal={arXiv preprint arXiv:2409.19812v4},
  year={2025}
}

@article{grunwald2024safe,
  title={Safe testing},
  author={Gr{\"u}nwald, Peter and de Heide, Rianne and Koolen, Wouter},
  journal={Journal of the Royal Statistical Society. Series B: Statistical Methodology},
  volume={86},
  number={5},
  pages={1091--1128},
  year={2024},
  publisher={Oxford University Press}
}

@article{sheskin1991computing,
  title={Computing absorption probabilities for a Markov chain},
  author={Sheskin, Theodore J},
  journal={International Journal of Mathematical Education in Science and Technology},
  volume={22},
  number={5},
  pages={799--805},
  year={1991},
  publisher={Taylor \& Francis}
}

@article{zhang2024humor,
  title={Humor in ai: Massive scale crowd-sourced preferences and benchmarks for cartoon captioning},
  author={Zhang, Jifan and Jain, Lalit and Guo, Yang and Chen, Jiayi and Zhou, Kuan L and Suresh, Siddharth and Wagenmaker, Andrew and Sievert, Scott and Rogers, Timothy and Jamieson, Kevin and others},
  journal={Advances in Neural Information Processing Systems},
  volume={37},
  pages={125264--125286},
  year={2024}
}

@inproceedings{yang2017,
 author = {Yang, Fanny and Ramdas, Aaditya and Jamieson, Kevin and Wainwright, Martin J},
 booktitle = {Advances in Neural Information Processing Systems},
 editor = {I. Guyon and U. Von Luxburg and S. Bengio and H. Wallach and R. Fergus and S. Vishwanathan and R. Garnett},
 pages = {},
 publisher = {Curran Associates, Inc.},
 title = {A framework for Multi-A(rmed)/B(andit) Testing with Online FDR Control},
 volume = {30},
 year = {2017}
}

@inproceedings{jamieson2014lil,
  title={lil’ucb: An optimal exploration algorithm for multi-armed bandits},
  author={Jamieson, Kevin and Malloy, Matthew and Nowak, Robert and Bubeck, S{\'e}bastien},
  booktitle={Conference on Learning Theory},
  pages={423--439},
  year={2014},
  organization={PMLR}
}

@article{howard2021time,
  title={Time-uniform, nonparametric, nonasymptotic confidence sequences},
  author={Howard, Steven R and Ramdas, Aaditya and McAuliffe, Jon and Sekhon, Jasjeet},
  journal={The Annals of Statistics},
  volume={49},
  number={2},
  pages={1055--1080},
  year={2021},
  publisher={JSTOR}
}

@article{turner2023exact,
  title={Exact anytime-valid confidence intervals for contingency tables and beyond},
  author={Turner, Rosanne J and Gr{\"u}nwald, Peter D},
  journal={Statistics \& Probability Letters},
  volume={198},
  pages={109835},
  year={2023},
  publisher={Elsevier}
}

@article{romano2011consonance,
  title={Consonance and the Closure Method in Multiple Testing},
  author={Romano, Joseph P and Shaikh, Azeem and Wolf, Michael},
  journal={The International Journal of Biostatistics},
  volume={7},
  number={1},
  pages={1--27},
  year={2011}
}

@article{bretz2011graphical,
  title={Graphical approaches for multiple comparison procedures using weighted Bonferroni, Simes, or parametric tests},
  author={Bretz, Frank and Posch, Martin and Glimm, Ekkehard and Klinglmueller, Florian and Maurer, Willi and Rohmeyer, Kornelius},
  journal={Biometrical Journal},
  volume={53},
  number={6},
  pages={894--913},
  year={2011}
}

@article{marcus1976closed,
  title={On closed testing procedures with special reference to ordered analysis of variance},
  author={Marcus, Ruth and Eric, Peritz and Gabriel, K Ruben},
  journal={Biometrika},
  volume={63},
  number={3},
  pages={655--660},
  year={1976},
  publisher={Oxford University Press}
}

@article{holm1979simple,
  title={A simple sequentially rejective multiple test procedure},
  author={Holm, Sture},
  journal={Scandinavian journal of statistics},
  pages={65--70},
  year={1979}
}
\bibliographystyle{plainnat}

\begin{appendices}

\section{Proofs}\label{sec:proof-appendix}
\subsection{Proofs of results in Section \ref{sec:e-value-closed-testing}}

\begin{proof}[\textbf{Proof of Theorem \ref{th:post-hoc-fwer}}]

We can upper bound this quantity by replacing $\hat\alpha$ with $\alpha$ and taking the supremum, and moreover, by Lemma \ref{lem:key}, we can write out \begin{align*}
    \BE\left[\frac{\mathbf{1}\{\max_{i\in\mathcal{H}_0}e_i^*\geq\frac{1}{\hat\alpha}\}}{\hat\alpha}\right]\leq\BE\left[\sup_\alpha\frac{\mathbf{1}\{\max_{i\in\mathcal{H}_0} e_i^*\geq\frac{1}{\alpha}\}}{\alpha}\right]\leq\BE\left[\sup_\alpha\frac{\mathbf{1}\{e_{\mathcal{H}_0}\geq\frac{1}{\alpha}\}}{\alpha}\right].
\end{align*}

By \cite{grunwald2024beyond}, this right-hand quantity is bounded by 1 since $e_{\mathcal{H}_0}$ is a valid e-value and thus post-hoc valid with scaled error $\leq1$ in expectation. However, we can reconstruct a simple proof of this fact by seeing that $\mathbf{1}\{e_{\mathcal{H}_0}\geq\frac{1}{\alpha}\}/\alpha$ is maximized when the numerator is 1 and $\alpha$ is as large as possible, which happens when $\alpha=1/e_{\mathcal{H}_0}$ as long as $e_{\mathcal{H}_0}\geq1$. Otherwise, we can never reject. This computes the upper bound as \begin{align*}
    \BE\left[\sup_\alpha\frac{\mathbf{1}\{e_{\mathcal{H}_0}\geq\frac{1}{\alpha}\}}{\alpha}\right]=\BE\left[\mathbf{1}\{e_{\mathcal{H}_0}\geq1\}e_{\mathcal{H}_0}\right]\leq\BE[e_{\mathcal{H}_0}]\leq1.
\end{align*}
    
\end{proof}

\begin{proof}[\textbf{Proof of Theorem \ref{th:post-hoc-anytime-valid}}]

The proof follows from \eqref{eq:max-nulls} and the properties of $e_{\mathcal{H}_0,\tau}$ as an e-value, so we may bound

\begin{align*}
    \BE\left[\frac{\mathbf{1}\{\max_{i\in\mathcal{H}_0} e_{i\tau}^*\geq\frac{1}{\hat\alpha}\}}{\hat\alpha}\right]&\leq\BE\left[\frac{\mathbf{1}\{e_{\mathcal{H}_0,\tau}\geq\frac{1}{\hat\alpha}\}}{\hat\alpha}\right]\leq\BE\left[\sup_\alpha\frac{\mathbf{1}\{e_{\mathcal{H}_0,\tau}\geq\frac{1}{\alpha}\}}{\alpha}\right]\\&=\BE\left[\mathbf{1}\{e_{\mathcal{H}_0,\tau}\geq1\}e_{\mathcal{H}_0,\tau}\right]\leq\BE\left[e_{\mathcal{H}_0,\tau}\right]\leq1
\end{align*}
    
\end{proof}

\begin{proof}[\textbf{Proof of Theorem \ref{th:always-valid}}]

We first swap the supremum over $t$ and maximum over $i$, then appeal to the stochastic dominance \eqref{eq:max-nulls}. \begin{align*}
    \BP\left(\max_{i\in\mathcal{H}_0}\sup_te_{it}^*\geq\frac{1}{\alpha}\right)=\BP\left(\sup_t\max_{i\in\mathcal{H}_0}e_{it}^*\geq\frac{1}{\alpha}\right)\leq\BP\left(\sup_te_{\mathcal{H}_0,t}\geq\frac{1}{\alpha}\right)\leq\alpha,
\end{align*}
where the last inequality follows from Ville's inequality \eqref{eq:villes}.
    
\end{proof}

\subsection{Proofs of results in Section \ref{sec:weighted-closed-testing}}

\begin{Lemma}\label{lem:pareto-bound}
    Let $K\geq1$ and $\alpha\in(0,1)$. Suppose $(e_{it})_{t\geq0}$ are independent e-processes and define $\tilde e_i=\max_{t\geq0}e_{it}$. Let $Y_1,...,Y_K$ be independent $\textrm{Pareto}(1)$ random variables. Then
    \begin{align*}        \sup_{w\in\Delta_K}\mathbb{P}\left(\sum_{i=1}^Kw_i\tilde e_i\geq\frac{1}{\alpha}\right)\leq B_K(\alpha)
    \end{align*}
where $\Delta_K=\{w=(w_1,...,w_K):w_i\geq0,\sum_{i=1}^Kw_i=1\}$ is the $K$-dimensional simplex and  
    \begin{align*}
B_K(\alpha):=\mathbb{P}\left(\sum_{i=1}^KY_i\geq\frac{K}{\alpha}\right)=\int_{\substack{y_1,...,y_K\geq1 \\ y_1+\cdots+y_K\geq K/\alpha}}\prod_{i=1}^K\frac{dy_i}{y_i^2}.
    \end{align*}
    The maximizing weights are uniform $w_i=1/K$ and the bound is sharp.
\end{Lemma}

\begin{proof}

Ville's inequality implies $\tilde e_i\preceq Y_i$, where $\preceq$ denotes stochastic domination. Since the $\tilde e_i$'s are independent, we may couple independent copies such that $\tilde e_i\leq Y_i$ a.s. for each $i$. The event
\begin{align*}
    \left\{x\in[0,\infty)^K:\sum_{i=1}^Kw_ix_i\geq\frac{1}{\alpha}\right\}
\end{align*}

is increasing in each coordinate $x_i$, therefore

\begin{align*}
    \mathbb{P}\left(\sum_{i=1}^Kw_i\tilde e_i\geq\frac{1}{\alpha}\right)\leq\mathbb{P}\left(\sum_{i=1}^Kw_iY_i\geq\frac{1}{\alpha}\right).
\end{align*}

It remains to maximize the Pareto bound over \(w\). Let \(Y,Z\) be independent \(\mathrm{Pareto}(1)\) variables. For \(a,b\ge 0\), write \(m=a+b\). For \(x\le m\),
\[
\BP(aY+bZ\ge x)=1,
\]
because \(Y,Z\ge 1\). For \(x>m\), a direct integration gives
\[
\BP(aY+bZ\ge x)
=
\frac{m}{x}
+
\frac{ab}{x^2}
\log\left(1+\frac{x(x-m)}{ab}\right),
\]
with the second term interpreted as \(0\) when \(ab=0\). For fixed \(m\) and \(x>m\), the only dependence on \((a,b)\) is through \(p=ab\). Since
\[
\frac{d}{dp}\left\{p\log\left(1+\frac{C}{p}\right)\right\}
=
\log\left(1+\frac{C}{p}\right)-\frac{C}{p+C}>0,
\qquad C:=x(x-m)>0,
\]
the probability above is increasing in \(ab\). For fixed \(a+b=m\), the product \(ab\) is maximized at \(a=b=m/2\). Thus replacing any two weights \((a,b)\) by their average \((m/2,m/2)\) cannot decrease the Pareto tail probability.

Now condition on all Pareto variables except two coordinates. The preceding two-variable argument applies with the conditional threshold obtained after subtracting the contribution of the other coordinates. Hence any pairwise averaging of two unequal weights cannot decrease
\[
\BP\left(\sum_{i=1}^K w_iY_i\ge \frac{1}{\alpha}\right).
\]
Repeated pairwise averaging, equivalently majorization, shows that the maximum over \(\Delta_K\) is attained at the uniform vector. This proves
\[
\sup_{w\in\Delta_K}
\BP\left(\sum_{i=1}^K w_iY_i\ge \frac{1}{\alpha}\right)
=
\BP\left(\frac{1}{K}\sum_{i=1}^K Y_i\ge \frac{1}{\alpha}\right).
\]

Finally, sharpness follows from the independent e-processes
\[
e_{i,t}=\exp(B_{i,t}-t/2),\qquad t\ge 0,
\]
where \(B_1,\ldots,B_K\) are independent standard Brownian motions. For each \(i\),
\[
\BP\left(\sup_{t\ge 0} e_{i,t}\ge y\right)=\frac{1}{y},
\qquad y\ge 1,
\]
so the running maxima are independent \(\mathrm{Pareto}(1)\) variables. Taking uniform weights attains \(B_K(\alpha)\).

\end{proof}

\begin{Lemma}\label{lem:pareto-recursion}
    Let \(T_k(s):=\BP(Y_1+\cdots+Y_k\ge s)\), where \(Y_1,\ldots,Y_k\) are independent \(\mathrm{Pareto}(1)\) variables. Then
\[
B_K(\alpha)=T_K(K/\alpha).
\]
The functions \(T_k\) are given exactly by the recursion
\[
T_1(s)=\frac{1}{s},\qquad s\ge 1,
\]
and, for \(k\ge 2\) and \(s\ge k\),
\[
\boxed{
T_k(s)
=
\frac{1}{s-k+1}
+
\int_1^{s-k+1}\frac{1}{y^2}\,T_{k-1}(s-y)\,dy .
}
\]
Thus the bound
\[
\sup_{w\in\Delta_K}
\BP\left(\sum_{i=1}^K w_i\tilde e_i\ge \frac{1}{\alpha}\right)
\le B_K(\alpha)
\]
is fully nonasymptotic for every fixed \(K\) and \(\alpha\in(0,1)\).
\end{Lemma}

\begin{proof}
    The identity \(B_K(\alpha)=T_K(K/\alpha)\) is the definition of \(B_K\). For the recursion, condition on \(Y_1=y\). Since \(Y_2+\cdots+Y_k\ge k-1\) almost surely, the event \(Y_1+\cdots+Y_k\ge s\) is certain whenever \(y\ge s-k+1\). Hence, for \(s\ge k\),
\[
T_k(s)
=
\int_1^{s-k+1}\frac{1}{y^2}T_{k-1}(s-y)\,dy
+
\int_{s-k+1}^{\infty}\frac{dy}{y^2}.
\]
The second integral equals \(1/(s-k+1)\), giving the stated recursion.
\end{proof}

\begin{proof}[\textbf{Proof of Theorem \ref{th:fwer-pareto-bound}}]

Let \(T_k(s)=\BP(Y_1+\cdots+Y_k\ge s)\). We prove by induction that, for each fixed \(k\ge 1\),
\[
T_k(s)
=
\frac{k}{s}
+
\frac{k(k-1)\log s}{s^2}
+O_k\left(\frac{1}{s^2}\right),
\qquad s\to\infty .
\]
The case \(k=1\) is exact. Assume the result holds for \(k-1\). From the nonasymptotic recursion,
\[
T_k(s)
=
\frac{1}{s-k+1}
+
\int_1^{s-k+1}\frac{1}{y^2}T_{k-1}(s-y)\,dy .
\]
Split the integral at \(s/2\). On \(1\le y\le s/2\), the induction hypothesis gives, uniformly in \(y\),
\[
T_{k-1}(s-y)
=
\frac{k-1}{s-y}
+
\frac{(k-1)(k-2)\log(s-y)}{(s-y)^2}
+O_k\left(\frac{1}{(s-y)^2}\right).
\]
The elementary estimates
\[
\int_1^{s/2}\frac{dy}{y^2(s-y)}
=
\frac{1}{s}+\frac{\log s}{s^2}+O\left(\frac{1}{s^2}\right)
\]
and
\[
\int_1^{s/2}\frac{\log(s-y)}{y^2(s-y)^2}\,dy
=
\frac{\log s}{s^2}+O\left(\frac{1}{s^2}\right)
\]
therefore imply
\[
\int_1^{s/2}\frac{1}{y^2}T_{k-1}(s-y)\,dy
=
\frac{k-1}{s}
+
\frac{(k-1)^2\log s}{s^2}
+O_k\left(\frac{1}{s^2}\right).
\]

For the remaining part, put \(u=s-y\). Since \((s-u)^{-2}=s^{-2}+O(u/s^3)\) for \(u\le s/2\), and since \(T_{k-1}(u)\le C_k/u\) for \(u\ge k-1\),
\[
\int_{s/2}^{s-k+1}\frac{1}{y^2}T_{k-1}(s-y)\,dy
=
\frac{1}{s^2}\int_{k-1}^{s/2}T_{k-1}(u)\,du
+O_k\left(\frac{1}{s^2}\right).
\]
The induction hypothesis also gives
\[
\int_{k-1}^{s/2}T_{k-1}(u)\,du
=(k-1)\log s+O_k(1).
\]
Hence
\[
\int_{s/2}^{s-k+1}\frac{1}{y^2}T_{k-1}(s-y)\,dy
=
\frac{(k-1)\log s}{s^2}
+O_k\left(\frac{1}{s^2}\right).
\]
Finally,
\[
\frac{1}{s-k+1}=\frac{1}{s}+O_k\left(\frac{1}{s^2}\right).
\]
Combining the three displays yields
\[
T_k(s)
=
\frac{k}{s}
+
\frac{\{(k-1)^2+(k-1)\}\log s}{s^2}
+O_k\left(\frac{1}{s^2}\right),
\]
which is the desired formula because \((k-1)^2+(k-1)=k(k-1)\).

Taking \(s=K/\alpha\) gives
\[
B_K(\alpha)
=
\alpha+
\frac{K-1}{K}\alpha^2\log\frac{K}{\alpha}
+O_K(\alpha^2)
=
\alpha+
\frac{K-1}{K}\alpha^2\log\frac{1}{\alpha}
+O_K(\alpha^2),
\]
because \(K\) is fixed.
    
\end{proof}

\subsection{Proofs of results in Section \ref{sec:e-holm}}

\begin{proof}[\textbf{Proof of Theorem \ref{th:eholm_rejection}}]
We can rewrite the rejection rule \eqref{eq:ebonf} as
\begin{equation}\label{eq:eholm_step1}
e_i \ge  \frac{1}{\alpha}+\max_{I\ni i}\sum_{j\in I\setminus\{i\}}\left(\frac{1}{\alpha}-e_{j}\right).
\end{equation}
Taking $I = \{i\}$ implies that $H_i$ would not be rejected when $e_i < 1/\alpha$. Thus we assume $e_i \ge 1/\alpha$ throughout the rest of the proof. It is clear that the RHS of \eqref{eq:eholm_step1} is bounded by $1/\alpha + C$ from above. On the other hand, it equals $1/\alpha + C$ if $I = \{i\}\cup \{j: e_j < 1/\alpha\}$. The proof is then completed. 
\end{proof}

\subsection{Proofs of results in Section \ref{sec:efall}}

\begin{proof}[\textbf{Proof of Lemma \ref{lem:remove_larger_than_i}}]
For any $I\ni i$, $J=I\cap[i]$ has $\alpha e_I\ge \alpha e_J$, since the $\alpha$-budget for $j>i$ gets assigned nowhere. Specifically, taking $i=i_{\ell'}$ without loss of generality, 
\begin{equation*}
\alpha e_I-\alpha e_J=\sum_{\ell=\ell'+1}^k\left(\sum_{i=i_{\ell-1}+1}^{i_\ell}\alpha_i\right)e_{i_\ell}\geq0.
\end{equation*}
\end{proof}

\begin{proof}[\textbf{Proof of Theorem \ref{th:efall-pf}}]
By Lemma \ref{lem:remove_larger_than_i}, we only need to show that $m_i=\min_{i\in I\subset[i]}\alpha e_I$. For each $i=i_k$, we consider the various options for $i_{k-1}$, which can be any element of $[i-1]$ or $k=1$ so $I=\{i\}$. Thus we can rewrite $$\alpha e_i^{*} = \min_{i\in I\subset[i]}\alpha e_I=\min_{0\leq j<i}\left[\min_{j\in J\subset[j]}\alpha e_J+e_i\sum_{k=j+1}^i\alpha_k\right],$$ where the minimum for $j=0$ is 0 by convention. Then we use induction, with our base case of $m_1=\alpha_1e_1$, where we assume that $\min_{j\in J\subset[j]}\alpha e_J=m_j$ for all $j\leq i-1$, in which case we get that $\min_{i\in I\subset[i]}\alpha e_I=m_i$, as intended.

\end{proof}

\begin{proof}[\textbf{Proof of Theorem \ref{th:efallback_optimal_I}}]
We prove the result by induction. The case for $i = 1$ is evident by Lemma \ref{lem:remove_larger_than_i}. Suppose the result holds for $1, \ldots, i-1$. If $j(i) = 0$, $e_i > e_j$ for all $j < i$ and $I_i^* = \{i\}$. Then, for any $i\in I \subset [i]$,
\[\alpha e_{I}\ge \left(\sum_{j=1}^{i}\alpha_j\right) e_i = \alpha e_{I_i^*}.\]
For the rest of the proof, we assume $j(i) > 0$. 

Let $I_i = \{i_1, \ldots, i_k\}$, with $i_1 < i_2 < \ldots < i_k = i$, be any subset such that $e_{i}^{*} = e_{I_i}$. We first prove that there exists $I_i'$ such that $e_i^* = e_{I_i'}$ and $j(i)$ is the second largest element. Assume that $j(i)\in (i_{\ell - 1}, i_{\ell}]$ for some $\ell \le k$, let $I_i' = \{i_1, \ldots, i_{\ell - 1}, j(i), i\}$. Then 
\begin{align*}
&\alpha e_{I_i} - \alpha e_{I_i'} = \sum_{r=\ell}^k\left(\sum_{j=i_{r-1}+1}^{i_r}\alpha_j\right)e_{i_{r}} - \left(\sum_{j=i_{\ell-1}+1}^{j(i)}\alpha_j\right)e_{j(i)} - \left(\sum_{j=j(i)+1}^{i}\alpha_j\right)e_{i}\\
& = \left(\sum_{j=i_{\ell-1}+1}^{j(i)}\alpha_j\right)(e_{i_{\ell}} - e_{j(i)}) + \left(\sum_{j=j(i)+1}^{i_\ell}\alpha_j\right)(e_{i_\ell} - e_i) + \sum_{r=\ell+1}^k\left(\sum_{j=i_{r-1}+1}^{i_r}\alpha_j\right)(e_{i_{r}} - e_i).
\end{align*}
By definition of $j(i)$, $e_{i_r} > e_i \ge e_{j(i)}$ for any $r\ge \ell$. Thus, $e_{I_i}\ge e_{I_i'}$. Since $I_i$ is a minimizer, $I_i'$ must be so as well. 

Since we can assume $j(i)$ is the second largest element in $I_i$, $e_i^*$ can be written as
\[e_i^* = \frac{1}{\alpha}\left(\sum_{j=j(i)+1}^i \alpha_j\right) e_i + \min_{j(i)\in J\subset [j(i)]} e_J.\]
By the induction hypothesis, the second term is $e_{I_{j(i)}^*}$. Thus, 
\[e_{i}^{*} = \frac{1}{\alpha}\left(\sum_{j=j(i)+1}^i \alpha_j\right) e_i + e_{j(i)}^*.\]
By construction, $e_i^* = e_{I_i^*}$. 
\end{proof}

\begin{proof}[\textbf{Proof of Lemma \ref{lem:iterative_Ii*}}]
If $j(i) = 0$, the result follows because $e_i > e_j$ for all $j < i$, $I_{i}^{*} = \{i\}$, and $I_{i-1}^{*}\setminus \{j\in I_{i-1}^{*}: e_j > e_i\} = \emptyset$. If $j(i) > 0$, we are left to show that
\begin{equation}\label{eq:efallback_goal}
I_{j(i)}^* =  \{k_1(i-1), \ldots, k_{m'_i}(i-1)\}, \quad \text{where }m'_i = \max\{m: e_{k_{m}(i-1)} \le e_{i}\}.
\end{equation}
By definition of $j(i)$, for any $j(i) < k < i$, $e_{k} > e_{i} \ge e_{j(i)}$. Thus, for each $j\le j(i)$,
\[\min\{e_{j}, \ldots, e_{i-1}\} = \min\{e_{j}, \ldots, e_{j(i)}\}.\]
 \eqref{eq:efallback_goal} then follows from the representation \eqref{eq:backward_cummin}.
\end{proof}

\subsection{Proofs of results in Section \ref{sec:dag}}

\begin{proof}[\textbf{Proof of Lemma \ref{lem:e_j^I}}]
    The proof is a straightforward manipulation of the expansion of $e_I$ using the definition of $w_i(I)$, regrouping by the $\alpha_j$ terms. 

    \begin{align*}
        e_I&=\frac{1}{\alpha}\sum_{i\in I}\left(\alpha_i+\sum_{j\notin I}\alpha_j\sum_{p\in\mathcal{P}_{j,i}^{(I)}}\prod_{\ell=1}^{k(p)}q_{i_{\ell-1},i_\ell}\right)e_i\\
        &=\frac{1}{\alpha}\sum_{i\in I}\alpha_ie_i+\frac{1}{\alpha}\sum_{i\in I}\left(\sum_{j\notin I}\alpha_j\sum_{p\in\mathcal{P}_{j,i}^{(I)}}\prod_{\ell=1}^{k(p)}q_{i_{\ell-1},i_\ell}\right)e_i\\
        &=\frac{1}{\alpha}\sum_{i\in I}\alpha_ie_i^{(I)}+\frac{1}{\alpha}\sum_{j\notin I}\alpha_j\left(\sum_{i\in I}e_i\sum_{p\in\mathcal{P}_{j,i}^{(I)}}\prod_{\ell=1}^{k(p)}q_{i_{\ell-1},i_\ell}\right)\\
        &=\frac{1}{\alpha}\sum_{i\in I}\alpha_ie_i^{(I)}+\frac{1}{\alpha}\sum_{j\notin I}\alpha_je_j^{(I)}=\frac{1}{\alpha}\sum_{j=1}^n\alpha_je_j^{(I)}.
    \end{align*}
\end{proof}

\begin{proof}[\textbf{Proof of Lemma \ref{lem:e_j^I-rec}}]
    To prove this, we divide the set of paths $\mathcal{P}_{j,i}^{(I)}$ into those with a given first node $i_1=k$ which is a child of $j$. These will all have the same first term $q_{j,k}$. Then we know that the resulting set of paths is the same as the set $\mathcal{P}_{k,i}^{(I)}$, but with an additional index $j$ at the start. In the case where the first node is $i_1=i$, this trivially reduces to $q_{j,i}e_i=q_{j,i}e_i^{(I)}$. Then we can rewrite the definition of $e_j^{(I)}$ for $j\notin I$ as 

    \begin{align*}
        e_j^{(I)}&=\sum_{i\in I}e_i\left(\sum_{p\in\mathcal{P}_{j,i}^{(I)}}\prod_{\ell=1}^{k(p)}q_{i_{\ell-1},i_\ell}\right)\\
        &=\sum_{(j,k)\in E}q_{j,k}\sum_{i\in I}e_i\left(\sum_{p\in\mathcal{P}_{j,i}^{(I)},i_1=k}\prod_{\ell=2}^{k(p)}q_{i_{\ell-1},i_\ell}\right)\\
        &=\sum_{(j,k)\in E}q_{j,k}\sum_{i\in I}e_i\left(\sum_{p\in\mathcal{P}_{k,i}^{(I)}}\prod_{\ell=1}^{k(p)}q_{i_{\ell-1},i_\ell}\right)=\sum_{(j,k)\in E}q_{j,k}e_k^{(I)}.
    \end{align*}
\end{proof}

\begin{proof}[\textbf{Proof of Theorem \ref{th:dag-alg}}]
    We first prove the first part by induction backward from $j = n$ to $j = 1$, for any $i\in[n]$, fixing $I\ni i$. For the base case $j = n$, we consider two cases. If $i = n$, it is clear that $e_i^{(I)}\ge e_i = e_i^{(i)}$. If $i < n$, since the hypotheses are topologically ordered, $n\notin A_i$ and hence $e_n^{(i)} = 0\le e_n^{(I)}$.
    Assuming that $e_k^{(I)}\geq e_k^{(i)}$ for all $k>j$, and we wish to show that $e_j^{(I)}\geq e_j^{(i)}$. If $j\notin A_i$ we have it automatically as $e_j^{(i)}=0$, and otherwise, 

    \begin{align}\label{e_j^i-comp}
        e_j^{(i)}=\min\left(e_j,\sum_{(j,k)\in E}q_{jk}e_k^{(i)}\right)\leq e_j^{(I)},
    \end{align}
    as 
    \begin{align*}
        e_j^{(I)}=\begin{cases}e_j & j\in I \\ \sum_{(j,k)\in E}q_{jk}e_k^{(I)} & j\not\in I\end{cases},
    \end{align*}
    and $e_j^{(i)}\leq e_j$ and $e_j^{(i)}\leq\sum_{(j,k)\in E}q_{jk}e_k^{(i)}\leq\sum_{(j,k)\in E}q_{jk}e_k^{(I)}$ by the inductive assumption, meaning regardless of whether $j\in I$, $e_j^{(i)}\leq e_j^{(I)}$ as intended.

    For the second part, let 
    \[I_{i}^{(i)} = \{i\}, \quad I_{j}^{(i)} = \begin{cases}
   I_{j+1}^{(i)} \cup \{j\} & e_j^{(i)} = e_j, j\in A_i\\  
    I_{j+1}^{(i)} & \text{otherwise}
    \end{cases}.\]
    We then construct $I_i^* = I_1^{(i)}$. For any $j\notin A_i$, there is no path from $j$ to any element in $I$. By Definition \ref{def:e_j^I}, $e_j^{(I_i^*)} = 0$ for any $j > i$. We prove $e_j^{(i)} = e_j^{(I_i^*)}$ for $j\in A_i$ by induction backward from $j = i$ to $j = 1$, for each $i \in [n]$. For the base case $j = i$, $e_i^{(i)} = e_i$, which is $e_i^{(I_i^*)}$ by Definition \ref{def:e_j^I}. Assuming that $e_k^{(i)} = e_k^{(I_i^*)}$ for all $k\in A_i$ and $k > j$, we wish to prove $e_j^{(i)} = e_j^{(I_i^*)}$. If $j \in I_i^*$, $e_j^{(I_i^*)} = e_j$ by Definition \ref{def:e_j^I} and $e_j^{(i)} = e_j$ by construction of $I_i^*$. If $j\notin I_i^*$, 
    \[e_j^{(i)} =\sum_{(j,k)\in E}q_{j,k}e_k^{(i)}.\] 
    By the induction hypothesis, $e_k^{(i)} = e_k^{(I_i^*)}$ for all $k$ such that $(j, k)\in E$. By Lemma \ref{lem:e_j^I-rec},  \[e_j^{(i)} 
 = \sum_{(j,k)\in E}q_{j,k}e_k^{(I_i^*)} = e_j^{(I_i^*)}.\]
 The proof is then completed. 
\end{proof}

\section{Post-hoc FWER control on the running maxima of e-processes}\label{sec:always-post-hoc}

Theorem \ref{th:always-valid} no longer holds for general data-driven $\alpha$. Following the counterexample from the proof of Proposition 3.1 of \cite{tavyrikov2025carefree}, we construct two e-processes $(e_{1t},e_{2t})$ where $e_{10}=e_{20}=1$ and $e_{it}=\prod_{s=1}^te_s^i$, where \begin{align*}
    (e_s^1,e_s^2)=\begin{cases}
        (1/2,1/2) & \text{w.p. }1/3 \\
        (1/2,2) & \text{w.p. }1/3 \\
        (2,1/2) & \text{w.p. }1/3.
    \end{cases}
\end{align*}
This is the same construction as \cite{tavyrikov2025carefree} except the initial values are set to 1 instead of jointly set to $1/(2\alpha)$ with probability $2\alpha$ and 0 otherwise.

When run with the running maxima $\tilde e_{it}=\max_{0\leq s\leq t}e_{is}$, e-Holm, the closure of e-Bonferroni with equal weights, makes a rejection if $\tilde e_{1t}+\tilde e_{2t}\geq2/\alpha$ and either $\tilde e_{1t}\geq1/\alpha$ or $\tilde e_{2t}\geq1/\alpha$. In this way, for any $t>0$, we can set $\hat\alpha_t=2/(\tilde e_{1t}+\tilde e_{2t})$, which is guaranteed to be below $1$ as $\tilde e_{it}\ge e_{i0} = 1$, and we are guaranteed that at time $t$, at least one hypothesis, $i_t=\arg\max_{i\in\{1,2\}}\tilde e_{it}$, will be rejected at level $\hat\alpha_t$. In this case, the family-wise scaled risk is \begin{align*}\mathbb{E}\left[\frac{\mathbf{1}\{\text{reject a hypothesis at level }\hat\alpha_t\}}{\hat\alpha_t}\right]=\mathbb{E}\left[\frac{\tilde e_{1t}+\tilde e_{2t}}{2}\right]=\mathbb{E}[\tilde e_{1t}]\end{align*} by symmetry. It is clear that $\mathbb{E}[\tilde e_{1t}] \geq \mathbb{E}[e_{1t}] = 1$. For example, taking $t=1$, this procedure will have inflated post-hoc risk, since $\tilde e_{1t}\geq 1$, and with probability 1/3, $\tilde e_{11}=2$, so $\mathbb{E}[\tilde e_{11}]\geq\frac{4}{3}>1$.

Nonetheless, following \cite{tavyrikov2025carefree}, we can apply adjusters to running maxima $\tilde{e}_{It} = \max_{s\le t}e_{Is}$ for all intersection hypotheses to convert them to e-processes with respect to the same filtration \citep[e.g.][]{shafer2011test, koolen2014buy, choe2024combining}. The adjusted closed test achieves post-hoc always-valid FWER control. 

\begin{Th}\label{th:post-hoc-always-valid}
Let $A$ be any increasing function satisfying 
\[\int_1^{\infty}\frac{A(x)}{x^2}dx \le 1.\]
We define $\tilde e_i^*$ as the result of applying an e-value-based closed testing procedure to the adjusted maxima of e-processes $\tilde e_i=\max_{t\geq0}e_{it}$. Then applying the adjuster $A$ to $\tilde e_i^*$, yielding $A(\tilde e_i^*)$, controls FWER for any data-driven level $\hat{\alpha}$ measurable in $\bigcup_{t\ge 1}\mathcal{F}_t$, i.e., 
  \begin{align*}  &\BE\left[\frac{\mathbf{1}\{\max_{i\in\mathcal{H}_0} A(\tilde e_i^*)\geq\frac{1}{\hat\alpha}\}}{\hat\alpha}\right]\leq1.
    \end{align*}

\end{Th}

\begin{proof}
By Ville's inequality \eqref{eq:villes}, $1/\sup_t e_{It}$ is a valid p-value for the intersection hypothesis $H_I$. Let $B(x) = A(1/x)$. The condition on $A$ is equivalent to 
\[\int_0^1 B(y)dy \le 1.\]
Thus, $A(\sup_t e_{It}) = B(1/\sup_t e_{It})$ is a valid e-value. Applying the same argument as in Theorem \ref{th:post-hoc-fwer}, the proof is completed.  
\end{proof}

The integral formulation in Theorem \ref{th:post-hoc-always-valid} for a valid adjuster is from \cite{koolen2014buy}. According to \cite{choe2024combining} and \cite{wang2025anytime}, $\{A(e_{It}): I\subset [n]\}$ are global e-processes even if $\{e_{It}: I\subset [n]\}$ are local e-processes with respect to a coarser filtration than $\mathcal{F}_t$. In such cases,  Theorem \ref{th:post-hoc-anytime-valid} continues to hold.

In Theorem \ref{th:post-hoc-always-valid}, there are three operations at play, namely applying the maximum over time, adjusting via closed testing, and calibrating with the adjuster $A$. We apply $A$ last, but it is also valid to apply the adjuster $A$ first then compute the closure, as $A(e_{it})$ is a valid global e-process. 

In practice, most popular adjusters are concave, including $A(e)=\sqrt{e}-1$ and $A(e)=\frac{e-1-\log e}{(\log e)^2}$ \citep{tavyrikov2025carefree}. For the weighted e-Bonferroni local tests we define in Section \ref{sec:weighted-e-bonf}, the Jensen's inequality implies $A(\tilde e_I)=A(\sum_{i\in I}w_i(I)\tilde e_i)\geq\sum_{i\in I}w_i(I)A(\tilde e_i)$ for each subset $I$, so applying the adjuster after closure leads to higher power than the reverse. 

\section{Additional detail on the uniform mixture example in Section \ref{sec:e-holm-vs-p-holm}}\label{sec:e-vs-p-holm-details}

First, we claim that the optimal p-value for testing $H_0:X\sim\textrm{Unif}[0,1]$ against $H_1:X\sim\pi\cdot\textrm{Unif}[0,1]+(1-\pi)\cdot\textrm{Unif}[1-c,1]$ is $p(X)=1-X$. Lemma \ref{lem:neyman-pearson} proves this claim.

\begin{Lemma}\label{lem:neyman-pearson}
In this setting, the p-value
\[
p(X)=1-X
\]
is exact under the null and is optimal for the mixture alternative above. In particular, for every $\alpha\in[0,1]$, the test
\[
\phi_\alpha(X)=I\{1-X\le \alpha\}
=I\{X\ge 1-\alpha\}
\]
maximizes $\BE_1[\phi(X)]$ over all measurable tests $\phi:[0,1]\to[0,1]$ satisfying $\BE_0[\phi(X)]\le \alpha$.
\end{Lemma}

\begin{proof}
Under the null, $X$ is uniform on $[0,1]$, so $1-X$ is also uniform on $[0,1]$. Thus $p(X)=1-X$ is an exact p-value.

It remains to prove optimality. Let $f_0$ and $f_1$ be the null and alternative densities with respect to Lebesgue measure on $[0,1]$. Then $f_0(x)=1$ and
\[
f_1(x)
=
\pi+(1-\pi)c^{-1}I\{x\in[1-c,1]\}.
\]
Set
\[
H=[1-c,1],
\qquad
 a=\pi+\frac{1-\pi}{c},
\qquad
 b=\pi .
\]
Then $f_1=a$ on $H$, $f_1=b$ on $[0,1]\setminus H$, and $a\ge b$.

Consider any measurable test $\phi:[0,1]\to[0,1]$ with $\BE_0[\phi(X)]\le \alpha$. Since $f_0=1$,
\[
\int_0^1 \phi(x)\,dx \le \alpha.
\]
The power of $\phi$ under the alternative is
\[
\BE_1[\phi(X)]
= b\int_0^1 \phi(x)\,dx
  +(a-b)\int_H \phi(x)\,dx .
\]
Because $0\le \phi\le 1$ and $H$ has length $c$,
\[
\int_H \phi(x)\,dx
\le
\min\{\alpha,c\}.
\]
Therefore every level-$\alpha$ test satisfies
\[
\BE_1[\phi(X)]
\le
b\alpha+(a-b)\min\{\alpha,c\}.
\]

Now consider $\phi_\alpha(x)=I\{x\ge 1-\alpha\}$. It has null rejection probability exactly $\alpha$. If $\alpha\le c$, then its rejection region lies inside $H$, and therefore
\[
\int_H \phi_\alpha(x)\,dx=\alpha .
\]
If $\alpha>c$, then its rejection region contains all of $H$, and therefore
\[
\int_H \phi_\alpha(x)\,dx=c .
\]
In both cases,
\[
\int_H \phi_\alpha(x)\,dx=\min\{\alpha,c\}.
\]
Thus $\phi_\alpha$ attains the upper bound above. Hence it maximizes power among all tests of level at most $\alpha$ for every $\alpha\in[0,1]$. This proves the optimality of $p(X)=1-X$.
\end{proof}

\section{More detailed description of the simulations in Section \ref{sec:new-yorker}}\label{sec:simulation-appendix}

First we describe the formulation for the bandit sampling algorithm in the New Yorker data semi-synthetic simulation. For both the Holm and DAG simulations, we take $n$ captions and run a lilUCB sampling algorithm \citep{jamieson2014lil}. We first draw one observation for each caption to get an initial average estimate $\hat\mu_i^{(0)}$ for each caption $i$, then compute an upper confidence bound $\textrm{UCB}_i^{(0)}$ for each caption. With hyperparameter $\delta$ and $\kappa$, and $n_i^{(t)}$ the number of pulls for caption $i$, the upper confidence bound, using a Bernoulli bound, is 

\begin{align*}
    \textrm{UCB}_i^{(t)}=\sup\left\{q\in[\hat\mu_i^{(t)},1]:d_{\textrm{bern}}(\hat\mu_i^{(t)},q)\leq\frac{\beta_i(t,\delta)}{n_i^{(t)}}\right\}.
\end{align*}

$\beta_i$ is the same function as in \cite{jamieson2014lil}, which is 

\begin{align*}
    \beta_i(t,\delta)=\log(\kappa\log(2n_i^{(t)}/\delta)/\log(2)).
\end{align*}

Then the bandit sampling algorithm pulls a uniformly random arm with probability $\epsilon=0.4$ and a uniform arm from the set of maximal UCBs $\{i:i=\arg\max_{j\in[n]}\textrm{UCB}_j^{(t)}\}$.

For the Holm simulation that is testing the caption means, the e-value at time $s$ for the sampled arm is $\exp(\eta(Y_s-\mu_0)-\psi)$ where $\eta$ is a tuning parameter we set to $\eta=0.25$, $\mu_0=0.25$ is the null against which we are testing, and $\psi$ is a normalization parameter such that the e-value is valid. For this construction we rescale the votes in $\{1,2,3\}$ to be in $\{0,0.5,1\}$ to bound the data in $[0,1]$. In general one can choose $(\eta,\psi)$ predictably to improve power but we are illustrating the multiple testing procedure, not optimizing for power.

The normalization term $\psi$ has the form
\begin{align*}
    \psi=\log((1-\mu)\cdot\exp(-\eta\cdot\mu_0)+\mu\cdot\exp(\eta\cdot(1-\mu_0))).
\end{align*}
As with the UCB formula, this is chosen as a conservative term, which is based on testing a Bernoulli random variable.

For the DAG simulation where we have $n=48$ captions, we use the same bandit sampling algorithm, but instead we test the effects of each treatment through a vector of contrasts. Specifically, each caption is assigned a group in $\{0,1\}^3$ based on their semantic values for each of the three features. Ordering these groups by their binary encoding (000, 001, 010, 011, 100, 101, 110, 111), we can construct the matrix of contrasts that corresponds to the main and interaction effects $\beta$. We ignore the three-way interaction effect in our testing.

To obtain the vector of $\beta$'s $(\beta_1,\beta_2,\beta_3,\beta_{12},\beta_{13},\beta_{23})$, we multiply the vector $\mu$ of true group-level means by a matrix of contrasts $C$, so $\beta=C\mu$ where
\begin{align*}
    C=\begin{pmatrix}
-1 & 0 & 0 & 0 & 1 & 0 & 0 & 0 \\
-1 & 0 & 1 & 0 & 0 & 0 & 0 & 0 \\
-1 & 1 & 0 & 0 & 0 & 0 & 0 & 0 \\
1 & 0 & -1 & 0 & -1 & 0 & 1 & 0 \\
1 & -1 & 0 & 0 & -1 & 1 & 0 & 0 \\
1 & -1 & -1 & 1 & 0 & 0 & 0 & 0 
\end{pmatrix}.
\end{align*}

We track the estimate for $\hat\beta$ for each effect, and we track the group mean $\hat\mu$, weighted such that each caption is treated equally. For the balanced design where each group has the same number of captions, this contrast construction corresponds exactly to the linear regression population coefficients. To test the coefficients against a null of 0, we employ a hedged process betting against each side of 0 to get a two-sided test. For betting scale $\lambda=0.05$ we set our bet for parameter $j$ at time $t$ as
\begin{align*}    \psi^{(t)}_j=\hat\beta_j^{(t-1)}+C_{ja}(y-\hat\mu_a^{(t-1)})/(6\pi),
\end{align*}
where $i$ is the caption that was pulled and $a=a(i)$ is the caption's group. In essence $a$ is the arm that was pulled, with the captions being elements of each arm. $\pi$ is the probability of pulling caption $i$ based on the bandit sampling algorithm. The division by 6 reflects the number of captions per arm which is necessary to scale the estimate correctly.

This bet is valid through an AIPW construction where the bet has conditional mean zero under the null, but which results in growth for an alternative even when that arm is not pulled.

The e-processes for effect $j$ are formed by 
\begin{align*}    e_{jt}^{(+)}=\prod_{s=1}^t(1+\lambda\psi_j^{(s)}), \quad e_{jt}^{(-)}=\prod_{s=1}^t(1-\lambda\psi_j^{(s)}).
\end{align*}
Then the overall e-process is a fixed mixture of these e-processes
\begin{align*}
    e_{jt}=\frac{1}{2}\left(e_{jt}^{(+)}+e_{jt}^{(-)}\right)
\end{align*}
which is valid and covers both sides of the null. We choose a fixed $\lambda=0.05$ such that the betting factors are always nonnegative. This is easy to do for a predictable $\lambda_t$, as we have an exploration floor $\epsilon$ and bounded outcomes. This type of hedged capital process is in the style of \cite{waudby2024estimating}.

\section{Extensions to select non-DAGs}\label{app:idx-loc}

In this appendix we extend our result from Section \ref{sec:dag} to a larger class of graphs which we call Index-Local DAGs. The key motivating observation is that our argument in Section \ref{sec:dag} relies on the fact that for every index $i$, the graph over which we have to optimize is a DAG, with $i$ as its only leaf. Thus, we define:

\begin{Def}\label{def:ildag} A graph $\mathcal{G}=(\mathcal{V},\mathcal{E})$ is an \textit{Index-Local DAG (ILDAG)} if for every $i$, the graph $\mathcal{G}^{(i)}$ formed by removing $j\in V$ such that there is no path from $j$ to $i$, as well as removing any $(i,k)\in \mathcal{E}$, is a DAG.

\end{Def}

The motivation for this is that when computing $e_i^*$, we have shown we only need to consider subsets of the ancestor graph $I\subset A_i$, and moreover only subsets containing $i$. Since $i\in I$, we know the $\alpha$-budget at $i$ will never be redistributed, i.e. $e_i^{(i)}=e_i$ in the notation of Section \ref{sec:dag}, so we can remove the edges from $i$ and have the same minimization problem. At this point, if the remaining graph is a DAG then we have by the result of Section \ref{sec:dag} that the DAG algorithm, for each $i$, will produce correct adjusted e-values $\{e_i^*\}_{i=1}^n$. This can be done by performing the DAG algorithm's backwards search for each $i$ on its ancestor graph $A_i$, and ignoring any reverse edges $i\rightarrow j$.

The result is somewhat limited, but there are at least two interesting cases. The first is the cyclical Fallback graph, which is identical but for an additional edge $n\rightarrow1$, ensuring that we are always working with the full $\alpha$-budget. This graph clearly satisfies the requirement to be an ILDAG, since removing any one edge $(i\rightarrow i+1)$ breaks the cycle and forms a Fallback chain starting at $i+1$ and ending at $i$.

The other example is the gatekeeper procedure \citep[e.g.][]{bretz2009graphical} for two endpoints with a cycle between them, a visualization of which is in Figure \ref{fig:gatekeeper}.

\begin{figure}
    \centering
    \includegraphics[scale = 0.2]{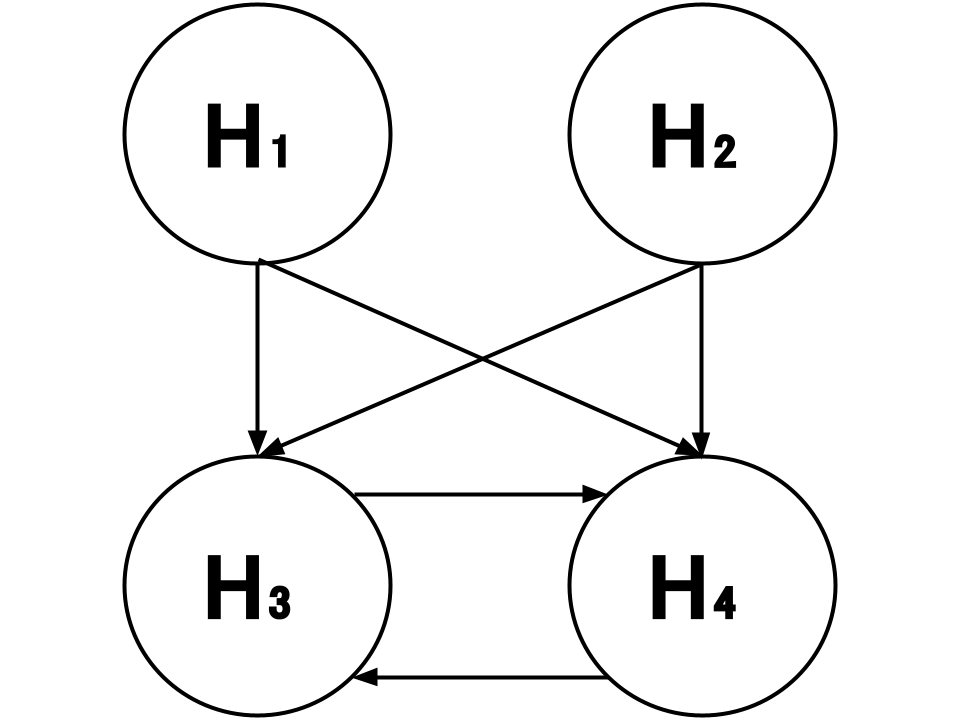}
    \caption{A gatekeeper graph with two nodes in each row. It is an example of an Index-Local DAG satisfying Definition \ref{def:ildag} since at least one of the two cycle-producing edges $H_3,H_4$ is removed when turning any given node into a leaf.}
    \label{fig:gatekeeper}
\end{figure}

Because we rerun Algorithm \ref{alg:dag} for each $i\in[n]$, with potentially all nodes included in $A_i$, a naïve upper bound on the worst-case runtime is $O(n^2|\mathcal{E}|)$.

\section{Expected runtime of e-Fallback algorithm \ref{alg:rev-search}}\label{app:exp-run-app}

\begin{Th}\label{th:exp-run} In the context of Algorithm \ref{alg:rev-search}:

\begin{enumerate}
    \item If the e-values $\{e_i\}_{i=1}^n$ are exchangeable, then the runtime is $O_P(n)$. 

    \item If there are $k$ non-null e-values, and the $n-k$ null e-values are exchangeable, then the runtime is $O_P(kn)$.
\end{enumerate}
    
\end{Th}

\begin{proof}
First we show part (1), by denoting $t_i$ as the number of comparisons to $e_i$ queried, i.e. $t_i=i-j(i)$ in the notation of the previous proof. Then we see that the runtime is $O(\sum_{i=1}^nt_i)$, where in expectation we compute 

\begin{align}
\BE_Pt_i=\sum_{k=1}^\infty\BP(t_i\geq k)=\sum_{k=1}^{i-1}\BP(t_i\geq k),
\end{align}
where $\BP(t_i\geq k)=\BP(\min(\{e_{i-k+1},...,e_i\}=e_i)\leq\frac{1}{k!}$ with equality if the distribution of $e_i$ is continuous, a result following from order statistics of exchangeable collections of random variables. By the Darth Vader Rule, $\BE_P[t_i]\leq\sum_{k=1}^{i-1}\frac{1}{k!}\leq e-1$, so we conclude that our runtime is $O_P((e-1)n)=O_P(n)$.

For part (2), we use a similar framework, but lower bound $t_i\leq n$ for $e_i$ a non-null, and $t_i\leq k+t_i'$ where $t_i'$ is the number of nulls preceding $i$ until we reach a null e-value with $e_j\leq e_i$. We enumerate the nulls by indices $\{i_j\}_{j=1}^{n-k}$, and define for each $i=i_\ell$ $j_0:=\max\{j<\ell|e_{i_j}\leq e_i\}$, the corresponding index for just the nulls. Then defining $t_i'=\ell-j_0$ we have $\BE_P[t_i']\leq e-1$ again, and our upper bound $t_i\leq k+e-1$. Putting these together for $k$ non-nulls and $n-k$ nulls, we have runtime 
\begin{align*}
    O_P(k(n)+(n-k)(k+e-1))&=O_P(nk+nk+(e-1)n-k^2-(e-1)k)\\&=O_P(nk+n),
\end{align*} so $O_P(nk)$ for $k>0$.

\end{proof}

This result is nice but uses exchangeability assumptions that may be unrealistic. However, it gives a sense of why we expect that this algorithm will provide a significant runtime improvement, which is achieved explicitly with the modified algorithm in Section \ref{sec:stack}.

\section{E-Holm adjusted e-values algorithm}\label{app:e-Holm} While the procedure in Theorem \ref{th:eholm_rejection} for the rejection set has an $O(n)$ complexity, it does not compute adjusted e-values. In this appendix, we re-derive and restate the e-Holm algorithm given by \cite{vovk2021values, vovk2023confidence}. By \eqref{eq:ebonf}, 
\[e_i^* = \min_{I\ni i}\frac{1}{|I|}\sum_{j\in I}e_j = \min_{I\ni i}\frac{1}{|I|}\left\{e_i + \sum_{j\in I\setminus\{i\}}e_j\right\}.\]
Clearly, the minimum is achieved by including in $I$ the $|I| - 1$ smallest e-values other than $e_i$. To pin down the exact expression, we sort our e-values $e_{(1)},...,e_{(n)}$ in decreasing order (to match the order of the associated 1/e p-values), and call $H_{(i)}$ the hypothesis associated with the e-value $e_{(i)}$ and adjusted e-value $e_{(i)}^*$. Then 

\begin{align*}
\min_{I\ni i,|I|=k+1}\frac{1}{|I|}\sum_{j\in I}e_{(j)}=\begin{cases}\frac{e_{(i)}+\sum_{j=n-k+1}^ne_{(j)}}{k+1} & k\leq n-i \\ \frac{\sum_{j=n-k}^ne_{(j)}}{k+1} & k>n-i
\end{cases}.
\end{align*}
To simplify the notation, let 
\[E_k = \sum_{j=n-k+1}^{n}e_{(j)}.\]
Then we can express the adjusted e-value for $H_{(i)}$ as 
\begin{equation}\label{eq:eholm_adjusted_e}
e_{(i)}^{*} = \min \left\{\min_{k\le n-i}\frac{e_{(i)} + E_k}{k+1}, \min_{k> n-i}\frac{E_{k+1}}{k+1}\right\} = \min_{k\le n-i}\frac{e_{(i)} + E_k}{k+1},
\end{equation}
where the second equality follows because $E_k / k$ is non-decreasing and thus, for any $k > n-i$,
\begin{equation*}
\min_{k > n-i}\frac{E_{k+1}}{k+1} \ge \frac{E_{n-i+1}}{n-i+1} = \frac{e_{(i)} + E_{n-i}}{n-i+1} \ge \min_{k\le n-i}\frac{e_{(i)} + E_k}{k+1}.
\end{equation*}
The terms $E_1, \ldots, E_n$ can be jointly computed in $O(n\log n)$ time, including sorting and cumulative sums. Calculating all adjusted e-values separately would result in $O(n^2)$ time in total because $e_{(i)}^{*}$ is the minimum of $n-i$ terms. We can reduce the computation cost to $O(n)$ by leveraging the following observation.
\begin{Lemma}\label{lem:ki}
Let $k_i = \max\{k: e_{(i)}^{*} = (e_{(i)} + E_{k})/(k + 1)\}$. Then $k_i$ is non-increasing in $i$.
\end{Lemma}

\begin{proof}
By definition of $k_i$, for any $k > k_i$, 
\[\frac{e_{(i)} + E_k}{k+1} > \frac{e_{(i)} + E_{k_i}}{k_i+1}\Longrightarrow \left( \frac{1}{k_i+1} - \frac{1}{k+1}\right) e_{(i)} <\frac{E_k}{k+1} - \frac{E_{k_i}}{k_i + 1}.\]
Since $e_{(i)}\ge e_{(i+1)}$, 
\[\left( \frac{1}{k_i+1} - \frac{1}{k+1}\right) e_{(i+1)}< \frac{E_k}{k+1} - \frac{E_{k_i}}{k_i + 1}\Longrightarrow \frac{e_{(i+1)} + E_k}{k+1} >  \frac{e_{(i+1)} + E_{k_i}}{k_i+1}.\]
This implies the minimum of $(e_{(i+1)} + E_{k})/(k + 1)$  cannot be attained for some $k > k_i$. As a result, $k_{i+1}\le k_i$.
\end{proof}

Lemma \ref{lem:ki} suggests $k_i$ and the adjusted e-values can be computed sequentially in a decreasing order from $i = n$ to $1$. The details are presented in Algorithm \ref{alg:ebonf}.

This algorithm correctly computes our minima as earlier argued, and has runtime $O(n\log n)$, from the sorting. The computation only changes $k$'s value $n$ times, and changes $e^*$ a maximum of $2n$ times.

\begin{algorithm}
    \textit{Input:} Vectors $\{e_i\}_{i=1}^n$;

    \textit{Sort:} $\{e_i\}_{i=1}^n$ from highest to lowest as $\{e_{(i)}\}_{i=1}^n$;

    \textit{Compute:} $E_1=e_{(n)}$, $E_k=E_{k-1}+e_{(n-k+1)}$ for $k=2$ to $k=n$;
    
    \textit{Initialize:} $e_{(n)}^*=e_{(n)}$, $e^*=e_{(n)},k=1$;
    
    \textit{For:} $i=n-1$ to 1,
    
    \textit{Do:} $e^*=e^*+\frac{e_{(i)}-e_{(i+1)}}{1+k}$;

    \indent While ($e^*>e_{(n-k)}$): $e^*=\frac{1+k}{2+k}e^*+\frac{1}{2+k}e_{(n-k)}$; $k=k+1$.

    \indent $e_{(i)}^*=e^*$.
    
    \textit{Output:} $\{e_{(i)}^*\}_{j=i}^n$.

    \caption{E-Holm adjusted e-values}\label{alg:ebonf}

\end{algorithm}

\end{appendices}

\end{document}